\documentclass[12pt]{article}
\usepackage{times}
\usepackage{graphicx}
\usepackage{amsfonts}
\usepackage{amsmath, amsthm, amssymb}
\usepackage{dsfont}
\usepackage{ucs}
\usepackage[utf8x]{inputenc}
\usepackage{color}
\usepackage{hyperref}
\usepackage{ifthen}
\title{Hebbian learning of recurrent connections: a geometrical perspective}
\author{Mathieu N. Galtier \footnote{Corresponding author: mathieu.galtier@inria.fr. NeuroMathComp Project Team, INRIA Sophia-Antipolis M\'editerran\'ee, 2004 route des Lucioles-BP 93, 06902 Sophia Antipolis, France} \and Olivier D. Faugeras \footnote{NeuroMathComp Project Team, INRIA Sophia-Antipolis M\'editerran\'ee, 2004 route des Lucioles-BP 93, 06902 Sophia Antipolis, France} \and Paul C. Bressloff \footnote{Department of Mathematics, University of Utah, 155 South 1400 East, Salt Lake City, Utah 84112, USA. Mathematical Institute, University of Oxford, 24-29 St. Giles', Oxford OX1 3LB, UK}}

\def \Z {\mathcal{Z}}
\def \C {\mathbb{C}}
\def \H {\mathcal{H}}
\def \I {\mathcal{I}}

\def \H {\mathcal{H}}

\def \O {\mathcal{O}}
\def \R {\mathbb{R}}
\def \N {\mathbb{N}}

\def \T {T} 

\def \eps {\epsilon}

\newtheorem{thm}{Theorem}[section]

\newboolean{DisplaySOS}
\setboolean{DisplaySOS}{false} 
\newcommand{\SOS}[1]{\ifthenelse{\boolean{DisplaySOS}}{{\textcolor{red}{\bf[#1]}}}{}}

\begin{document}
 \maketitle
\paragraph{Abstract:}
We show how a Hopfield network with modifiable recurrent connections undergoing slow Hebbian learning can extract the underlying geometry of an input space.
First, we use a slow/fast analysis to derive an averaged system whose dynamics derives from an energy function and therefore always converges to equilibrium points. The equilibria reflect the correlation structure of the inputs, a global object extracted through local recurrent interactions only. Second, we use numerical methods to illustrate how learning extracts the hidden geometrical structure of the inputs. Indeed, multidimensional scaling methods make it possible to project the final connectivity matrix on  to a distance matrix in a high-dimensional space, with the neurons labelled by spatial position within this space. The resulting network structure turns out to be roughly convolutional. The residual of the projection defines the non-convolutional part of the connectivity which is minimized in the process. Finally, we show how restricting the dimension of the space where the neurons live gives rise to patterns similar to cortical maps. We motivate this using an energy efficiency argument based on wire length minimization. Finally, we show how this approach leads to the emergence of ocular dominance or orientation columns in primary visual cortex. In addition, we establish that the non-convolutional (or long-range) connectivity is  patchy, and is co-aligned in the case of orientation learning.

\section{Introduction}
Activity-dependent synaptic plasticity is generally thought to be the basic cellular substrate underlying learning and memory in the brain. Donald Hebb \cite{hebb:49} postulated that learning is based on the correlated activity of synaptically connected neurons: if both neurons A and B are active at the same time, then the synapses from A to B and B to A should be strengthened proportionally to the product of the activity of A and B. However, as it stands, Hebb's learning rule diverges. Therefore, various modification of Hebb's rule have been developed, which basically take one of three forms (see \cite{gerstner-kistler:02} and \cite{dayan-abbott:01}): first, a decay term can be added to the learning rule so that each synaptic weight is able to ``forget'' what it previously learned. Second, each synaptic modification can be normalized or projected on different subspaces. These constraint--based rules may be interpreted as implementing some form of competition for energy between dendrites and axons, see \cite{miller:96, miller-mackay:96} and \cite{ooyen:01} for details. Third, a sliding threshold mechanism can be added to Hebbian learning. For instance, a post-synaptic threshold rule consists in multiplying the presynaptic activity and the subtraction of the average postsynaptic activity from its current value, which is referred as covariance learning (\cite{sejnowski1989hebb}). Probably the best known of these rules is the BCM rule \cite{bienenstock-cooper-etal:82}. It should be noted that history-based rules can also be defined without changing the qualitative dynamics of the system: instead of considering the instantaneous value of the neurons' activity, these rules consider its weighted mean over a time window (see \cite{foldiak1991learning,wallis1997optimal}). Recent experimental evidence suggests that learning may also depend upon the precise timing of action potentials \cite{bi2001synaptic}. Contrary to most Hebbian rules that only detect correlations, these rules can also encode causal relationships in the patterns of neural activation. However, the mathematical treatment of these spike timing dependent rules is much more difficult than rate based ones.

Hebbian-like learning rules have often been studied within the framework of unsupervised feedfoward neural networks \cite{oja:82, bienenstock-cooper-etal:82,miller-mackay:96,dayan-abbott:01}. They also form the basis of most weight-based models of cortical development, assuming fixed lateral connectivity (e.g. mexican hat) and modifiable vertical connections (see the review of \cite{swindale1996development})\footnote{There have only been a few computational studies that consider the joint development of lateral and vertical connections \cite{bartsch2001combined,miikkulainen-bednar-etal:05}.}. In these developmental models, the statistical structure of input correlations provides a mechanism for spontaneously breaking some underlying symmetry of the neuronal receptive fields leading to the emergence of feature selectivity. When such correlations are combined with fixed intracortical interactions, there is a simultaneous breaking of translation symmetry across cortex leading to the formation of a spatially periodic cortical feature map. A related mathematical formulation of cortical map formation has been developed in \cite{takeuchi-amari:79,bressloff:05} using the theory of self--organizing neural fields. Although very irregular, the two-dimensional cortical maps observed at a given stage of development, can be unfolded in higher dimensions to get smoother geometrical structures. Indeed, \cite{bressloff-cowan-etal:01} suggested that the network of orientation pinwheels in V1 is a direct product between a circle for orientation preference and a plane for position, based on a modification of the icecube model of Hubel and Wiesel \cite{hubel77}. From a more abstract geometrical perspective, Petitot \cite{petitot:03} has associated such a structure to a 1-jet space and used this to develop some applications to computer vision. More recently, \cite{bressloff-cowan:03} and \cite{chossat2009hyperbolic} have considered more complex geometrical structures  such as spheres and hyperbolic surfaces that incorporate additional stimulus features such as spatial frequency and textures, respectively.

In this paper, we show how geometrical structures related to the distribution of inputs can emerge through unsupervised Hebbian learning applied to recurrent connections in a rate-based Hopfield network. Throughout this paper, the inputs are presented as an external non-autonomous forcing to the system and not an initial condition as is often the case in Hopfield networks. It has previously been shown that, in the case of a single fixed input, there exists an energy function that describes the joint gradient dynamics of the activity and weight variables \cite{dong1992dynamic}. This implies that the system converges to an equilibrium during learning. We use averaging theory to generalize the above result to the case of multiple inputs, under the adiabatic assumption that Hebbian learning occurs on a much slower time scale than both the activity dynamics and the sampling of the input distribution. We then show that the equilibrium distribution of weights, when embedded  into $\R^k$ for sufficiently large integer $k$, encodes the geometrical structure of the inputs. Finally, we numerically show that the embedding of the weights in two dimensions ($k=2$) gives rise to patterns that are qualitatively similar to experimentally observed cortical maps, with the emergence of features columns and patchy connectivity.
Although the mathematical formalism we introduce here could  be extended to most of the rate-based Hebbian rules in the literature, we present the theory for Hebbian learning with decay because of the simplicity of the resulting dynamics.

Note that the use of geometrical objects to describe the emergence of connectivity patterns has previously been put forward by Amari in a different context. Based on the theory of information geometry, Amari considers the geometry of the set of all the networks and defines learning as a trajectory on this manifold for perceptron networks in the framework of supervised learning \cite{amari1998natural} or for unsupervised Boltzmann Machines \cite{amari1992information}. He uses  differential and Riemannian geometry to describe an object which is at a larger scale than the cortical maps this paper is focusing on.

Moreover, Zucker and colleagues are currently developing a non-linear dimensionality reduction approach to caracterize the statistics of natural visual stimuli (see \cite{zucker_abstract,coifman2005geometric}). Although they do not use learning neural networks and stay closer to the field of computer vision than this paper, it turns out their approach is similar to the geometrical embedding approach we are using.

The structure of the paper is as follows. In section \ref{part: Model}, we formally introduce the  model. We derive the averaged system in section \ref{part: averaging}, which then allows us to study the stability of the learning dynamics in the presence of multiple inputs by constructing an appropriate energy function. We adress stability in section \ref{part: stability}. In section \ref{part: geometry of equilibrium point} we determine the geometrical structure of the equilibrium weight distribution and show how it reflects the structure of the inputs. We also relate this approach to the emergence of cortical maps. Finally, the results are discussed in section \ref{part: discussion}.

\section{Model}
\label{part: Model}
\subsection{Neural network evolution}
A neural mass corresponds to a mesoscopic coherent group of neurons. It is convenient to consider them as building blocks, first for computational simplicity, second for their direct relationship to macroscopic measurements of the brain (EEG, MEG and Optical imaging) which average over numerous neurons, and third because one can functionally define coherent groups of neurons within  cortical columns. For each neural mass $i \in \{1..N\}$, define the mean membrane potential $V_i(t)$ at time $t$. The instantaneous population  firing rate $\nu_i(t)$ is linked to the membrane potential through the relation $\nu_i(t) = s\big(V_i(t)\big)$, where $s$ is a smooth sigmoid function. In the following, we choose
\begin{equation}
 s(v) = \frac{S_m}{1+\exp\big(-4 S'_m(v-\phi)\big)},
\end{equation}
where $S_m$, $S'_m$ and $\phi$ are respectively the maximal firing rate, the maximal slope and the offset of the sigmoid.

Consider a Hopfield network of neural masses described by the equation
\begin{equation}
 \frac{dV_i}{dt}(t)=-\alpha V_i(t)+\sum_{j=1}^N W_{ij}(t)\ s\big(V_j(t)\big)+I_i(t) .
\end{equation}
The first term roughly corresponds to the intrinsic dynamics of the neural mass: it decays exponentially to zero at a rate $\alpha$ if it receives neither external inputs nor spikes from the other neural masses. We will fix the units of time by setting $\alpha=1$. The second term corresponds to the rest of the network sending information through spikes to the given neural mass $i$, with $W_{ij}(t)$ the effective synaptic weight from neural mass $j$. The synaptic weights are time--dependent because they evolve according to a continuous time Hebbian learning rule (see below). The third term $I_i(t)$ corresponds to an external input  to neural mass $i$, e.g. information extracted by the retina or thalamo-cortical connections. We take the inputs to be piecewise constant in time, that is, at regular time intervals a new input is presented to the network. In this paper, we will assume that the inputs are chosen by peridodically cycling through a given set of $M$ inputs. An alternative approach would be to randomly select each input from a given probability distribution \cite{geman79}. It is convenient to introduce vector notation by representing the time--dependent membrane potentials by $V \in C^1(\R_+, \R^N)$, the  time--dependent external inputs by $I \in C^0(\R_+, \R^N)$, and the time--dependent network weight matrix by $W \in C^1(\R_+, \R^{N \times N})$. We can then rewrite the above system of ordinary differential equations as a single vector-valued equation
\begin{equation}
 \frac{dV}{dt}=-V+W \cdot S(V)+I ,
 \label{eq: voltage based model}
\end{equation}
where $S: \R^N \rightarrow \R^N$ corrresponds to the term by term application of the sigmoid $S$, i.e. $S(V)_i = s(V_i)$. 

\subsection{Correlation-based Hebbian learning}
\label{part: correlation-based learning}
The synaptic weights are assumed to evolve according to a correlation--based Hebbian learning rule of the form
\begin{equation}
\frac{dW_{ij}}{dt} =\epsilon ( s(V_i)s(V_j)-\mu W_{ij}),
\label{Hebb}
\end{equation}
where $\epsilon$ is the learning rate, and we have included a decay term in order to stop the weights from diverging. In order to rewrite the above equation in a more compact vector form, we introduce the 
tensor (or Kronecker) product $S(V) \otimes S(V)$ so that in component form
\begin{equation}
[S(V) \otimes S(V)]_{ij} =  S(V)_i S(V)_j,
\label{eq: correlation-based learning}
\end{equation}
where $S$ is treated as a mapping from $\R^N$ to $\R^N$. The tensor product implements Hebb's rule that synaptic modifications involve the product of postynaptic and presynaptic firing rates. We can then rewrite the combined voltage and weight dynamics as the following non--autonomous (due to time--dependent inputs) dynamical system:
\begin{equation}
 \Sigma: \quad\left\{
    \begin{array}{lcl}
	{\displaystyle \frac{dV}{dt} }&=& -V+W \cdot S(V)+I\\  \\
      {\displaystyle  \frac{dW}{dt}}&= &\eps \Big(S(V) \otimes S(V) - \mu W \Big) .
    \end{array}
\right.
\label{eq: system Sigma}
\end{equation}

Let us make few remarks about the existence  and uniqueness of solutions. First, boundedness of $S$ implies boundedness of the system $\Sigma$. More precisely, if $I$ is bounded, then the solutions are bounded. To prove this, note that the right hand side of the equation for $W$ is the sum of a bounded term and a linear decay term in $W$. Therefore, $W$ is bounded and hence the term $W\cdot S(V)$ is also bounded. The same reasoning applies to $V$. $S$ being Lipschitz continuous implies that the right hand side of the system is Lipschitz. This is sufficient to prove existence and uniqueness of the solution by applying the Cauchy-Lipschitz theorem. In the following, we will derive an averaged autonomous dynamical system $\Sigma'$, which will be well-defined for the same reasons.

\section{Averaging the system}
\label{part: averaging}
We will show that system $\Sigma$ can be approximated by an autonomous Cauchy problem which will be much more convenient to handle. This averaging method makes the most of multiple time--scales in the system. First, it is natural to consider that learning occurs on a much slower time--scale than the evolution of the membrane potentials (as determined by $\alpha$), i.e.
\begin{equation}
 \eps \ll 1.
 \label{eq: time scales mu}
\end{equation}
Second, an additional time-scale arises from the rate at which the inputs are sampled by the network. That is, the network cycles periodically through $M$ fixed inputs, with the period of cycling given by $T$. It follows that $I$ is $T$--periodic, piecewise constant. We assume that the sampling rate is also much slower than the evolution of the membrane potentials,
\begin{equation}
\frac{M}{T} \ll 1.
 \label{eq: condition slow inputs}
\end{equation}
Finally, we assume that the period $T$ is small compared to the time-scale of the learning dynamics,
\begin{equation}
\epsilon \ll \frac{1}{T}.
 \label{av}
\end{equation}
We can now simplify the system $\Sigma$ by applying Tikhonov's theorem for slow/fast systems, and then classical averaging methods for periodic systems.

\subsection{Tikhonov's theorem}
Tikhonov's theorem (\cite{tikhonov1952systems} and \cite{verhulst2007singular} for a clear introduction) deals with slow/fast systems. It says the following:
\begin{thm}
 Consider the initial value problem
$$
\begin{array}{lcl}
	\dot{x} = f(x,y,t), \ x(0) = x_0,\ x\in \R^n, t \ \in \R_+\\
        \eps \dot{y} = g(x,y,t), \ y(0) = y_0,\ y\in \R^m\\
    \end{array}
$$
Assume that:
\begin{enumerate}
 \item A unique solution of the initial value problem exists and we suppose, this holds also for
the reduced problem
$$
\begin{array}{lcl}
	\dot{x} = f(x,y,t), \ x(0) = x_0 \\
        0 = g(x,y,t)
    \end{array}
$$
with solutions $\bar{x}(t)$, $\bar{y}(t)$.

 \item The equation $0 = g(x, y, t)$ is solved by $\bar{y}(t) = \phi(x, t)$, where $\phi(x, t)$ is a continuous function
and an isolated root. Also suppose that $\bar{y}(t) = \phi(x, t)$ is an asymptotically stable solution
of the equation $\frac{dy}{d\tau} = g(x,y,\tau)$  that is uniform in the parameters $x \in \R^n$ and $t\in \R_+$.
 \item $y(0)$ is contained in an interior subset of the domain of attraction of $\bar{y}$.
\end{enumerate}
Then we have
$$
\begin{array}{lcl}
	\lim_{\eps \rightarrow 0} x_\eps(t) = \bar{x}(t),\ 0\leq t \leq L \\
	\lim_{\eps \rightarrow 0} y_\eps(t) = \bar{y}(t),\ 0\leq d \leq t \leq L \\
    \end{array}
$$
with $d$ and $L$ constants independent of $\eps$.
\end{thm}

In order to apply Tikhonov's theorem directly to the system $\Sigma$, we first need to rescale time according to $t \rightarrow  \epsilon t$. This gives
\begin{eqnarray*}
\epsilon \frac{dV}{dt} &=& -V+W \cdot S(V)+I\\
        \frac{dW}{dt}&= &S(V) \otimes S(V) - \mu W  .
        \end{eqnarray*}
Tikhonov's theorem then implies that solutions of $\Sigma$ are close to solutions of the reduced system (in the unscaled time variable)
\begin{equation}
 \left\{
    \begin{array}{lcl}
	V(t) & = & W \cdot S\big(V(t)\big)+I(t)\\
        \dot{W} &= &\epsilon \Big ( S(V) \otimes S(V) - \mu W \Big ),
    \end{array}
\right.
\label{eq: system after Tikhonov}
\end{equation}
provided that the dynamical systems $\Sigma$ in equation (6), and equation (\ref{eq: system after Tikhonov}) are well defined. It is easy to show that both systems are Lipschitz because of the properties of $S$. Following \cite{faugeras-grimbert-etal:08}, we know that if 
\begin{equation}
 S'_m \|W\|<1,
 \label{eq: contracting condition}
\end{equation}
then there exists an isolated root $\bar{V} : \R_+ \rightarrow \R^N$ of the equation $V=  W \cdot S(V)+I$ and $\bar{V}$ is asymptotically stable. Equation (\ref{eq: contracting condition}) corresponds to the weakly connected case. Moreover, the initial condition belongs to the basin of attraction of this single fixed point. Note that we require $\frac{M}{T} \ll 1$ so that the membrane potentials have sufficient time to approach the equilibrium associated with a given input before the next input is presented to the network. In fact, this assumption make it reasonable to neglect the transient activity dynamics due to the switching between inputs.

\subsection{Periodic averaging}
The system given by equation (\ref{eq: system after Tikhonov}) corresponds to a differential equation for $W$ with $T$-periodic forcing due to the presence of $V$ on the right--hand side. Since $T \ll \epsilon^{-1}$, we can use classical averaging methods to show that solutions of (\ref{eq: system after Tikhonov}) are close to solutions of the following autonomous system on the time-interval $[0,\frac{1}{\eps}]$ (which we suppose large because $\eps << 1$)
\begin{equation}
 \Sigma_0: \quad\left\{
    \begin{array}{lcl}
	V(t) & = & W \cdot S(V(t))+I(t)\\ \\
       {\displaystyle \frac{dW}{dt}}&= &\eps \Big(\displaystyle\frac{1}{T}\displaystyle\int_0^\T S(V(s)) \otimes S(V(s)) ds - \mu W(t)\Big) .
    \end{array}
\right.
\label{eq: system Sigma 0}
\end{equation}
It follows that solutions of $\Sigma$ are also close to solutions of $\Sigma_0$. Finding the explicit solution $V(t)$ for each input $I(t)$ is difficult and requires fixed points methods, e.g. a Picard algorithm. Therefore, we will consider yet another system $\Sigma'$ whose solutions are  also close to $\Sigma_0$ and hence $\Sigma$. In order to construct $\Sigma'$ we need to introduce some additional notation.

Let us label the $M$ inputs by $I^{(a)}, a=1,\ldots,M$ and denote by $V^{(a)}$ the fixed point solution of the equation $V^{(a)}=W\cdot S(V^{(a)})+I^{(a)}$.
Given the periodic sampling of the inputs, we can rewrite (\ref{eq: system Sigma 0}) as
\begin{eqnarray}
\begin{array}{lcl}	V^{(a)} & = &  W\cdot S(V^{(a)})+I^{(a)}\\ \\
      {\displaystyle  \frac{dW}{dt}} &= &\eps \Big({\displaystyle \frac{1}{M}}\displaystyle\sum_{a=1}^M S(V^{(a)}) \otimes S(V^{(a)})- \mu W(t)\Big) .
        \end{array}
        \label{a1}
\end{eqnarray}
If we now introduce the ${N \times M}$ matrices ${\mathcal V}$ and ${\mathcal I}$ with components ${\mathcal V}_{ia}= V_i^{(a)}$ and ${\mathcal I}_{ia}= I_i^{(a)}$, then we can eliminate the tensor product and simply write (\ref{a1}) in the matrix form
\begin{eqnarray}
\begin{array}{lcl}	{\mathcal V} & = &  W\cdot S({\mathcal V})+{\mathcal I}\\ \\
       {\displaystyle \frac{dW}{dt}}&= &\eps \Big(\displaystyle\frac{1}{M}\displaystyle S({\mathcal V}) \cdot S({\mathcal V})^T- \mu W(t)\Big) ,
        \end{array}
        \label{a2}
\end{eqnarray}
where $S({\mathcal V})\in  \R^{N \times M}$ such that $[S({\mathcal V})]_{ia}=s({V}_i^{(a)})$.
A second application of Tikhonov's theorem (in the reverse direction) then establishes that solutions of the system $\Sigma_0$ (written in the matrix form (\ref{a2})) are close to solutions of the matrix system
\begin{equation}
 \Sigma': \quad\left\{
    \begin{array}{lcl}
	{\displaystyle\frac{d{\mathcal V}}{dt} }& = & -{\mathcal V} +  W \cdot S\big({\mathcal V}\big)+{\mathcal I}\\ \\
       {\displaystyle \frac{dW}{dt}}&= &\eps\Big(\displaystyle\frac{1}{M}S({\mathcal V})\cdot S({\mathcal V})^T  - \mu W(t)\Big)
    \end{array}
\right. 
\label{eq: system Sigma prime}
\end{equation}

In the remainder of the paper we will focus on the system $\Sigma'$ whose solutions are close to those of the original system $\Sigma$ provided condition (\ref{eq: contracting condition}) is satisfied, i.e. the network is weakly connected. Clearly, the fixed points $(V^*,W^*)$ of system $\Sigma$ satisfy $\|W^*\| \leq \frac{S_m^2}{\mu}$. Therefore, equation (\ref{eq: contracting condition}) says that if $\frac{S_m^2 S'_m }{\mu}<1$ then Tikhonov's theorem can be applied and systems $\Sigma$ and $\Sigma'$ can be reasonably considered as good approximations of each other. The advantage of the averaged system $\Sigma'$ is that is given by autonomous ordinary differential equations. Moreover, since it is Lipschitz continuous, it leads to a well-posed Cauchy problem. Finally, note that it is straighforward to extend our approach to time-functional rules (e.g. sliding threshold or BCM rules as described in \cite{bienenstock-cooper-etal:82}) which, in this new framework,  would be approximated by simple ordinary differential equations (as opposed to time-functional differential equations) provided $S$ is redefined appropriately.

\begin{figure}[ht]
 \centering
 \includegraphics[width=0.5\textwidth]{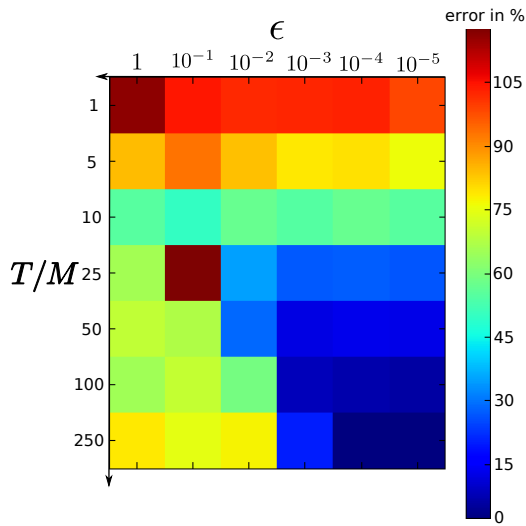}
 \caption{Percentage of error between final connectivities for the exact and averaged system.}
 \label{fig: averaging table}
\end{figure}

\begin{figure}[ht]
 \centering
 \includegraphics[width=0.5\textwidth]{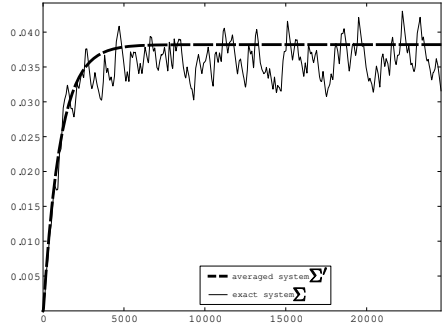}
\caption{Temporal evolution of the norm of the connectivities of the exact system $\Sigma$ and averaged system $\Sigma'$.} 
\label{fig: averaging evolution}
\end{figure}

\subsection{Simulations}
To illustrate the above approximation, we simulate a simple network with both exact, i.e. $\Sigma$, and averaged ,i.e. $\Sigma'$, evolution equations. For these simulations, the network consists of $N=10$ fully-connected neurons and is presented with $M=10$ different random inputs taken uniformly in the intervals $[0,1]^{N}$. For this simulation we use $s(x) = \frac{1}{1+e^{-4(x-1)}}$, and $\mu = 10$. Figure \ref{fig: averaging table}. shows the percentage of error between final connectivities for different values of $\eps$ and $T/M$. Figure \ref{fig: averaging evolution} shows the temporal evolution of the norm of the connectivity for both the exact and averaged system for $T = 10^3$ and $\eps = 10^{-3}$.

\section{Stability}
\label{part: stability}
\subsection{Liapunov function}
\label{part: Energy}
In the case of a single fixed input ($M=1$), the systems $\Sigma$ and $\Sigma'$ are equivalent and reduce to the neural network with adapting synapses previously analyzed by \cite{dong1992dynamic}. Under the additional constraint that the weights are symmetric ($W_{ij}=W_{ji}$), these authors showed that the simultaneous evolution of the neuronal activity variables and the synaptic weights can be re-expressed as a gradient dynamical system that minimizes a Liapunov or energy function of state. We can generalize their analysis to the case of multiple inputs ($M>1$) and non-symmetric weights using the averaged system $\Sigma'$. That is, following along similar lines to \cite{dong1992dynamic}, we introduce the energy function
\begin{equation}
E({\mathcal U},W) = -\frac{1}{2}\langle {\mathcal U},W\cdot {\mathcal U} \rangle - \langle {\mathcal I},{\mathcal U} \rangle +{\langle 1, \overline{S^{-1}}\big({\mathcal U}\big) \rangle} + \frac{M \mu}{2} \|W\|^2
\label{eq: energy learning}
\end{equation}
where ${\mathcal U}=S({\mathcal V})$, $ \|W\|^2=\langle W,W\rangle = \sum_{i,j}W_{ij}^2$,
\begin{equation}
\langle {\mathcal U},W\cdot {\mathcal U} \rangle=\sum_{a=1}^M\sum_{i=1}^N U_i^{(a)}W_{ij} U_j^{(a)},\quad 
\langle {\mathcal I},{\mathcal U} \rangle =\sum_{a=1}^M\sum_{i=1}^N I_i^{(a)}U_i^{(a)}
\end{equation}
and
\begin{equation}
 \langle 1, \overline{S^{-1}}\big({\mathcal U}\big) \rangle=\sum_{a=1}^M\sum_{i=1}^N \int_0^{U_i^{(a)}}S^{-1}(\xi)d\xi .
 \end{equation}
 In contrast to \cite{dong1992dynamic}, we do not require {\em a priori} that the weight matrix is symmetric. However, it can be shown that the system always converges to a symmetric connectivity pattern. More precisely, $\mathcal{A} = \Big\{({\mathcal V},W) \in \R^{N \times M} \times \R^{N \times N} :\ W= W^T\Big\}$ is an attractor of the system $\Sigma'$. A proof can be found in appendix \ref{part: appendix symmetric fixed points}. 
It can then be shown that on ${\mathcal A}$ (symmetric weights), $E$ is a Liapunov function of the dynamical system $\Sigma'$, that is, 
\begin{equation*}
\frac{dE}{dt} \leq 0,\quad \mbox{and} \quad \frac{dE}{dt}=0\implies \frac{d{\mathcal Y}}{dt}=0, \quad {\mathcal Y}=({\mathcal V},W)^T.
\end{equation*}
The boundedness of $E$ and the Krasovskii-LaSalle invariance principle then implies that the system converges to an equilibrium \cite{khalil-grizzle:96}. We thus have
\begin{thm}
 The initial value problem for the system $\Sigma'$ with ${\mathcal Y}(0) \in \H $, converges to an equilibrium state.
\label{thm: lasalle implies convergence}
\end{thm}
\begin{proof}
See appendix \ref{part: appendix Lasalle} 
\end{proof}
\noindent It follows that neither oscillatory nor chaotic attractor dynamics can occur.

\subsection{Linear stability}
\label{part: linear stability}
Although we have shown that there are stable fixed points, not all of the fixed points are stable. However, we can apply a linear stability analysis on the system $\Sigma'$ to derive a simple sufficient condition for a fixed point to be stable. The method we use in the proof could be extended to more complex rules. The proof reveals the significant role played by the Kronecker product in Hebbian learning.

\begin{thm}
The equilibria of system $\Sigma'$ satisfy:
\begin{equation}
 \left\{
    \begin{array}{lc}
        {\mathcal V}^* = \frac{1}{\mu M} S({\mathcal V}^*)\cdot S({\mathcal V}^*)^T \cdot S({\mathcal V}^*) +{\mathcal I}\\
	W^* = \frac{1}{\mu M} S({\mathcal V}^*) \cdot S({\mathcal V}^*)^T\\
    \end{array}
\right.
\label{eq: M0 fixed points subspace}
\end{equation}
and a sufficient condition for stability is
\begin{equation}
 3S'_m \|W^*\| < 1
\label{eq: sufficient condition M0}
\end{equation}
provided $1> \eps \mu$ which is most probably the case since $\eps<<1.$
\label{thm: fixed points and linear stability}
\end{thm}
\begin{proof}
 See appendix \ref{part: appendix linear stability}
\end{proof}
\noindent This condition is strikingly similar to that derived in \cite{faugeras-grimbert-etal:08}. In fact, condition (\ref{eq: sufficient condition M0}) is stronger than the contracting condition (\ref{eq: contracting condition}). It says the network may converge to a weakly connected situation. It justifies the averaging method by saying that we remain in the domain of validity of the averaging method. It also says that the dynamics of ${\mathcal V}$ is likely (because the condition is only sufficient) to be contracting and therefore subject to no bifurcations: a fully recurrent learning neural network is likely to have a ``simple'' dynamics.

\section{Geometrical structure of equilibrium points}
\label{part: geometry of equilibrium point}
\subsection{Learning the correlation matrix of the inputs}
\label{part: equilibrium point}
It follows from equation (\ref{eq: M0 fixed points subspace}) that the equilibrium weight matrix $W^*$ is given by the correlation matrix of the firing rates. Moreover, in the case of sufficiently large inputs, the matrix of equilibrium membrane potentials satisfies ${\mathcal V}^* \approx {\mathcal I}$. More precisely, if $|S(I_i^{(a)}) |\ll |I_i^{(a)}|$ for all $a =1,\ldots ,M$ and $i=1,\ldots,N$, then we can generate an iterative solution for ${\mathcal V}^*$ of the form
$$\mathcal{V}^* = {\I} + \frac{1}{\mu} S\big(\I\big) \cdot S\big(\I\big)^T \cdot S\big(\I\big) + {\rm h. o. t.}$$
On the other hand, if the inputs are comparable in size to the synaptic weights, then there is no explicit solution for $\mathcal{V}^*$. Roughly speaking, we observe that the connection term has the role of ``smoothing'' the solution. Therefore, if a Gaussian is presented to the network (as the only input), the membrane potential is likely to be another Gaussian with a larger variance.
If  no input is presented to the network ($I=0$), then $S(0) \neq 0$ implies that the activity is non-zero, that is, there is spontaneous activity. Combining these observations,  we see that the network roughly extracts and stores the correlation matrix of the strongest inputs within the weights of the network.

\subsection{From a symmetric connectivity matrix to a convolutional network}
\label{part: multidimensional scaling}
So far neurons have been identified by a label $i \in \{1..N\}$; there is no notion of geometry or space in the preceding results. However, as we show below, the inputs may contain a spatial structure that can be encoded by the connectivity. In this section, we propose a mechanism to unveil the hidden geometrical structure within the connectivity. More specifically, we want to find an integer $k \in \N$ and $N$ points in $\R^k$, denoted by $x_i, i\in \{1,\ldots,N\}$, so that the connectivity can roughly be written as $W^*_{ij} \simeq \exp(-\|x_i - x_j\|^2)$. In other words, we interpret the final connectivity as a matrix describing the distance between the neurons living in a k-dimensional space. However, $W^*$ is not always a distance matrix, therefore, it is natural to project it on the set of distance matrices. Finding the best fit of $W^*$ to a distance matrix is usually called multidimensional scaling. This set of methods is reviewed in \cite{borg2005modern}. 

First, define $\widehat{W} \in \R^{N\times N}$ as $W^*$ whose diagonal terms are set to $W^*_{max}$ the largest component of $W^*$: $W_{ij}^* = N_{ij} \widehat{W}_{ij}$ with $N_{ij} = 1$ if $i \neq j$ and $N_{ii} = W_{ii} / W^*_{max}$. Second, define a bijective kernel function on $x\in \R_+$ such that $K(x)=W^*_{max} e^{-x/\sigma^2}$. Given that $\widehat{W}$ is non-negative, we define the matrix $D = K^{-1}(\widehat{W})$ corresponding to the application of the inverse of $K$ to each component of $\widehat{W}$. As said before, we want to find $k \in \N$ and $x_i\in \R^k$ for $i\in \{1,\ldots,N\}$ so that $D$ is a distance matrix, $D_{ij} = \|x_i - x_j\|^2$. In general,  this is not possible. However, we can compute the projection of the symmetric matrix $\widehat{W}$ onto the set of distance matrices by applying multidimensional scaling methods as described in \cite{borg2005modern}. We use the stress majorization or SMACOF algorithm for the stress1 cost function throughout the article. In other words, we can find the distance matrix $D_\bot$ such that $\|D_\shortparallel\| = \|D - D_\bot\|$ is minimal. Therefore, $\widehat{W}_{ij} = K(D_\shortparallel)_{ij}K(D_\bot)_{ij}$. Define $M$ such that $M(x_i,x_j) = K(D_\shortparallel)_{ij}\ N_{ij}$ and $G_\sigma$ a Gaussian with a standard deviation equal to $\sigma$. In spatial coordinates
\begin{equation}
W^*(x_i,x_j) = M(x_i,x_j)\ G_\sigma(\|x_i - x_j\|)
\label{eq: connectivity decomposition}
\end{equation}
Multidimensional scaling methods consist in minimizing the contribution of $M$ in the preceding equation. Hence, we refer to $M$ as the non-convolutional connectivity.

A position $x_i \in \R^k$ is associated to each neuron $i \in \{1,\ldots,N\}$ such that
\begin{equation}
\big(W\cdot S(V)\big)_i = \sum_{j = 1}^N  W^*_{ij} S(V_j) = \sum_{j=1}^N M(x_i,x_j)\ G_\sigma(\|x_i - x_j\|)\ S(V(x_j))
\end{equation}
In particular, we can assume that $k$ is large enough for $\|D_\shortparallel\|$ to be very small. Moreover, if the neurons are equally excited on average (i.e. the diagonal of $W^*$ is already equal to $W_{max}^* I_d$), then it is reasonable to consider that $M(x_i,x_j) = 1$ leading to the following convolutional product
$$W\cdot S(V) = G_\sigma(\|.\|) \ast S(V)$$ 
Therefore, in the space defined by the $x_i \in \R^k$ the connectivity is close to being convolutional.
\subsection{Unveiling the geometrical structure of the inputs}
\label{part: simulations}
We hypothesize that the space defined by the $x_i$ reflects the underlying geometrical structure of the inputs. We have not found a way to prove this so we only provide numerical examples that illustrate this claim. In the following examples, we relate the geometry of the manifold suggested by the $x_i$ to the network inputs. Thus we feed the network with inputs having a defined geometrical structure and then show how this structure can be extracted from the connectivity by the method above. However, it is by extracting the structure from unknown inputs that these networks might reveal themselves useful. Therefore, the following is only a (numerical) proof of concept.

We assume the inputs to be uniformly distributed over a manifold $\Omega$ with fixed geometry. This strong assumption amounts to considering that the feedforward connectivity (which we do not consider here) has already properly filtered the information coming from the sensory organs. More precisely, define the set of inputs by the matrix $\I \in \R^{N \times M}$ such that $I_{i}^{(a)} = f(\|y_i - z_a\|_{\Omega})$ where the $z_a$ are uniformly distributed points over $\Omega$, the $y_i$ are the positions on $\Omega$ that ``label'' the $i$th neuron, and $f$ is a decreasing function on $\R_+$. The norm $\|.\|_\Omega$ is the natural norm defined over the manifold $\Omega$. For simplicity, assume $f(x) = f_\sigma(x)= A e^{-\frac{x^2}{\sigma^2}}$ so that the inputs are localized bell-shaped bumps on the shape $\Omega$.

\subsubsection{Planar retinotopy}

\begin{figure}[htbp]
 \centering
 \includegraphics[width=0.5\textwidth]{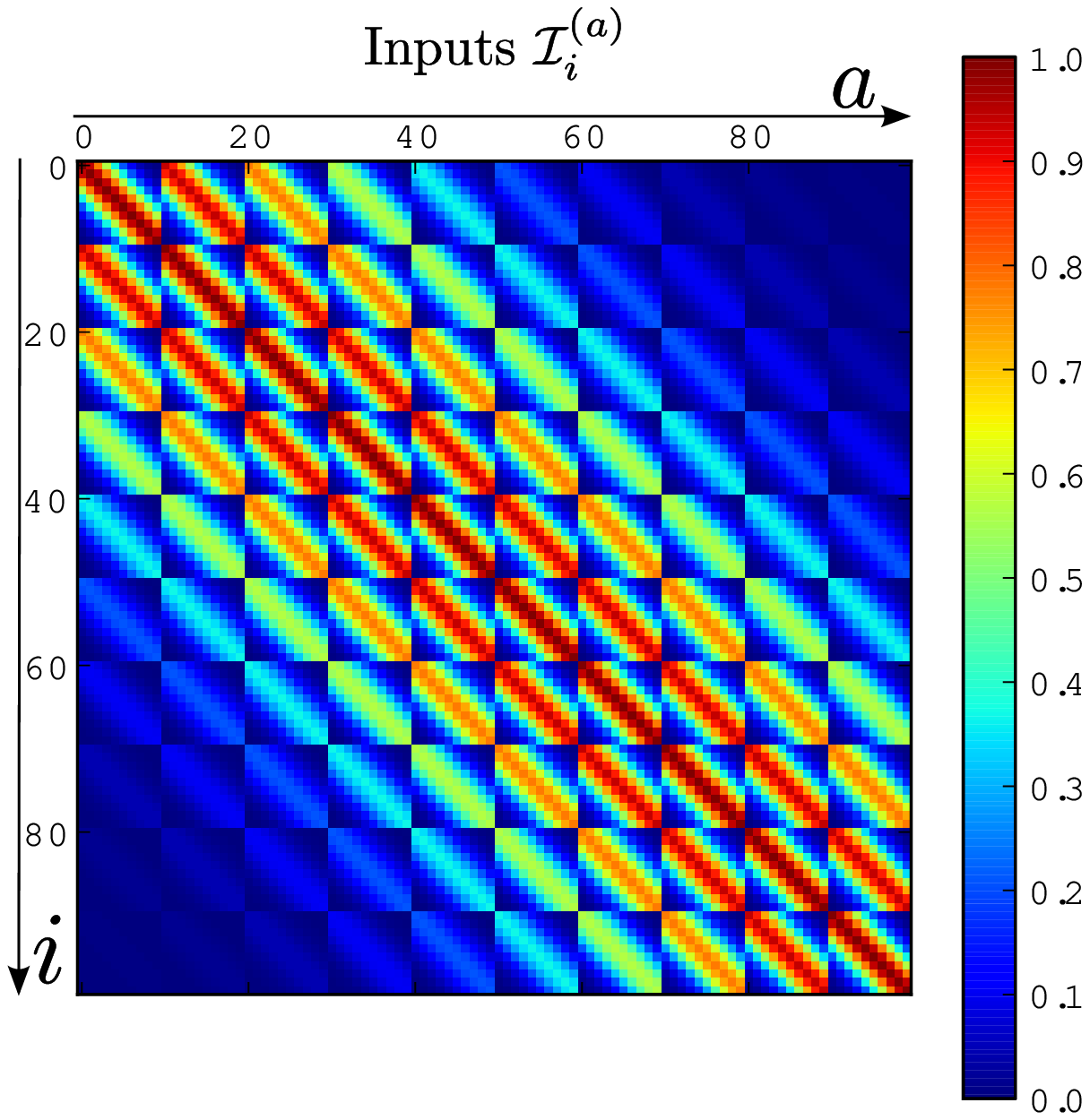}\hfill
 \includegraphics[width=0.5\textwidth]{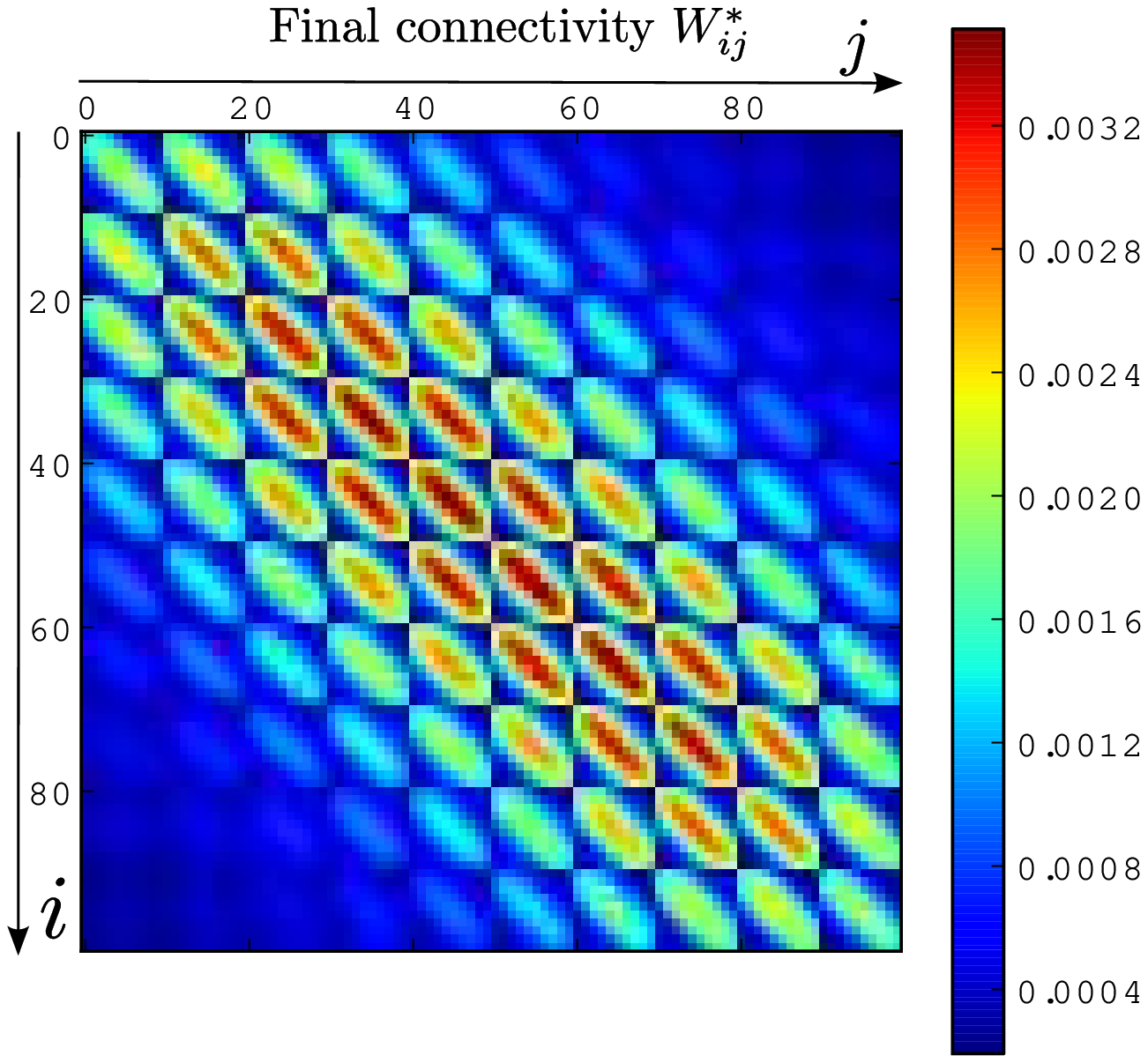}
 \caption{Plot of planar retinotopic inputs on $\Omega =  [0,1] \times [0,1]$ (left) and final connectivity matrix of the system $\Sigma'$ (right). The parameters used for this simulation are $s(x) = \frac{1}{1+e^{-4(x - 1)}}$, $l = 1$, $\mu = 10$, $\eps = 0.001$, $N=M = 100$, $\sigma = 4$.}
 \label{fig: 2d simus}
 \end{figure}

\noindent We consider a set of Gaussian inputs uniformly distributed over a two-dimensional plane, e.g. $\Omega = [0,1] \times [0,1]$. For simplicity, we take $N=M=K^2$ and set $z_a = y_i$ for $i=a$, $a \in \{1,\ldots,M\}$. (The numerical results show an identical structure for the final connectivity when the $y_j$ correspond to random points, but the analysis is harder). In the simpler case of one-dimensional Gaussians with $N=M=K$, the input matrix takes the form $\I=T_{f_{\sigma}}$, where $T_f$ is a symmetric Toeplitz matrix:
\begin{equation}
 T_{f} = 
\begin{pmatrix}
  f(0) & f(1) & f(2) & \cdots & \cdots & f(K) \\
  f(1) & f(0) & f(1) & f(2) & \cdots & f(K-1)\\
  f(2) & f(1) & f(0) & f(1) & \cdots & f(K-2)\\
  \vdots & \vdots & \ddots & \ddots & \ddots & \vdots \\
  f(K) & f(K-1) & f(K-2) & \cdots & \cdots & f(0) \\
 \end{pmatrix}
\end{equation}
In the two-dimensional case, we set $y=(u,v)\in \Omega$ and introduce the labeling $y_{k+(l-1)K}=(u_k,v_l)$ for $ k,l =1,\ldots K$. It follows that $I_i^{(a)}\sim \exp(-(u_k-u_{k'})^2)\exp(-(v_l-v_{l'})^2)$ for $i=k+(l-1)K$ and $a=k'+(l'-1)K$. Hence, we can write $\I = T_{f_\sigma} \otimes T_{f_\sigma}$, where $\otimes$ is the Kronecker product; the Kronecker product is responsible for the $K\times K$ sub-structure we can observe in figure~\ref{fig: 2d simus} with $K=10$. Note that if we were interested in a n-dimensional retinotopy, then the input matrix could be written as a Kronecker product between n Toeplitz matrices. As previously mentioned, the final connectivity matrix roughly corresponds to the correlation matrix of the input matrix. It turns out that the correlation matrix of $\I$ is also a Kronecker product of two Toeplitz matrix generated by a single Gaussian (with a different standard deviation). Thus, the connectivity matrix has the same basic form as the input matrix when $z_a = y_i$ for $i=a$. The inputs and stable equilibrium points of the simulated system are shown in figure~\ref{fig: 2d simus}. The positions $x_i$ of the neurons after multidimensional scaling are shown in figure~\ref{fig: retinotopy positions}.

 \begin{figure}[htbp]
 \centering
 \includegraphics[width=0.45\textwidth]{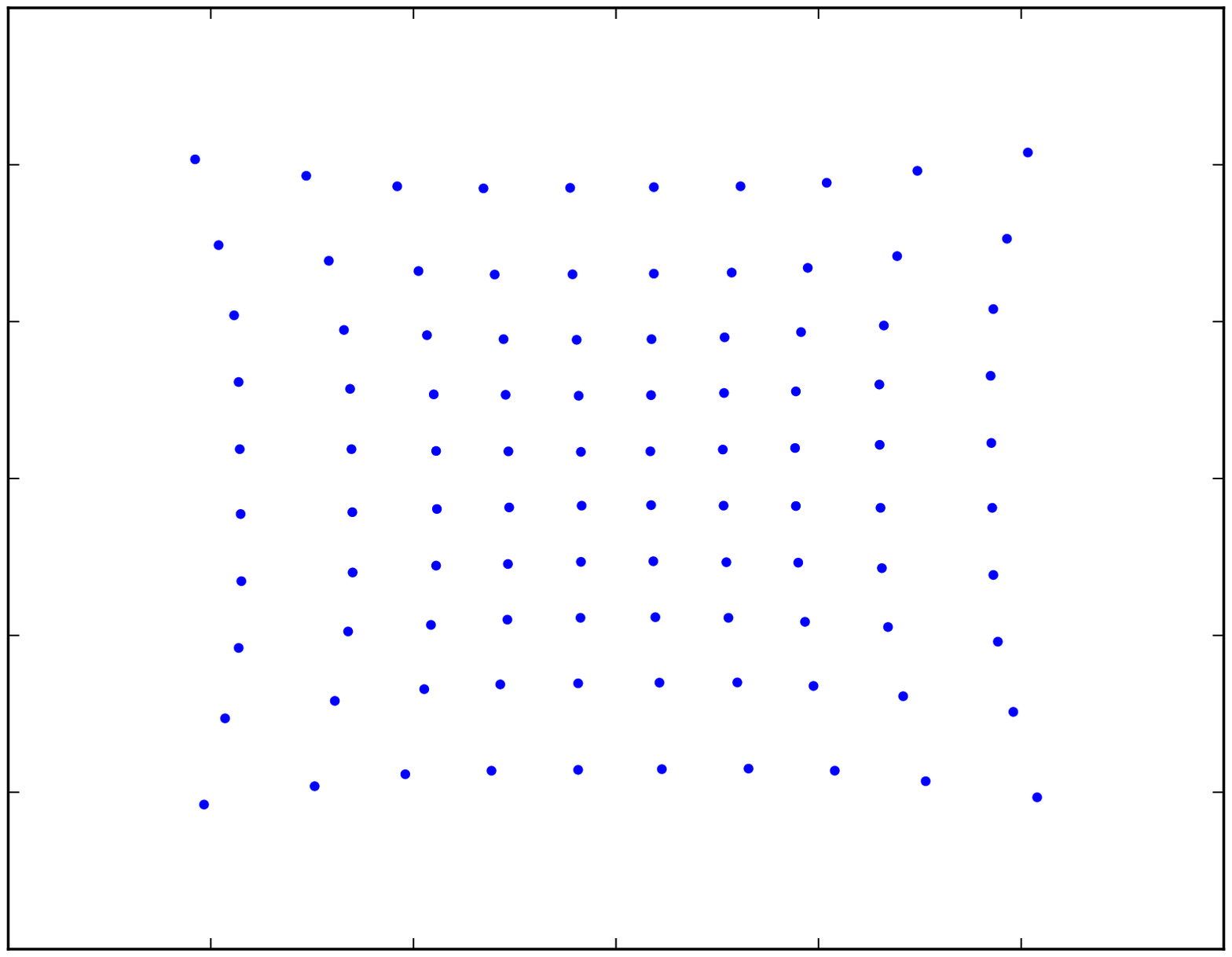}\hfill
 \includegraphics[width=0.45\textwidth]{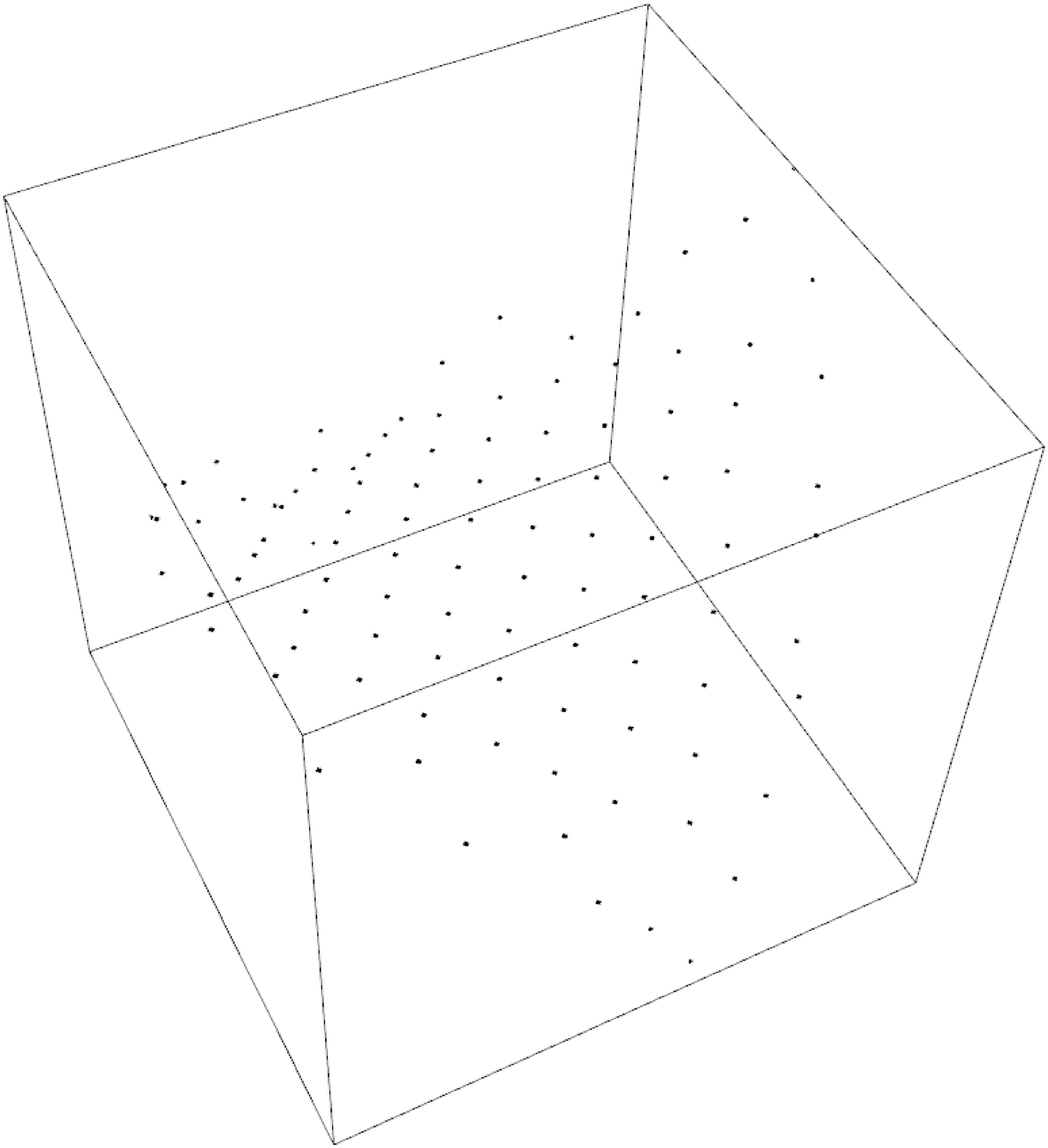}
 \caption{Positions $x_i$ of the neurons after having applied classical multidimensional scaling to the final connectivity matrix shown in figure \ref{fig: 2d simus} for $k=2$ (left) and $k=3$ (right). The regular spacing of the neurons for $k=2$ shows that the planar structure of the inputs has been recovered, although the corner of the square appear less regular due to boundary effects. In the case $k = 3$, there is an embedding of the plane into three dimensions; the saddle--like shape accounts for the corner irregularity observed when $k=2$. }
 \label{fig: retinotopy positions}
\end{figure}

\subsubsection{Toro\"idal retinotopy}
\label{part: torus}
\begin{figure}[htbp]
 \centering
 \includegraphics[width=0.55\textwidth]{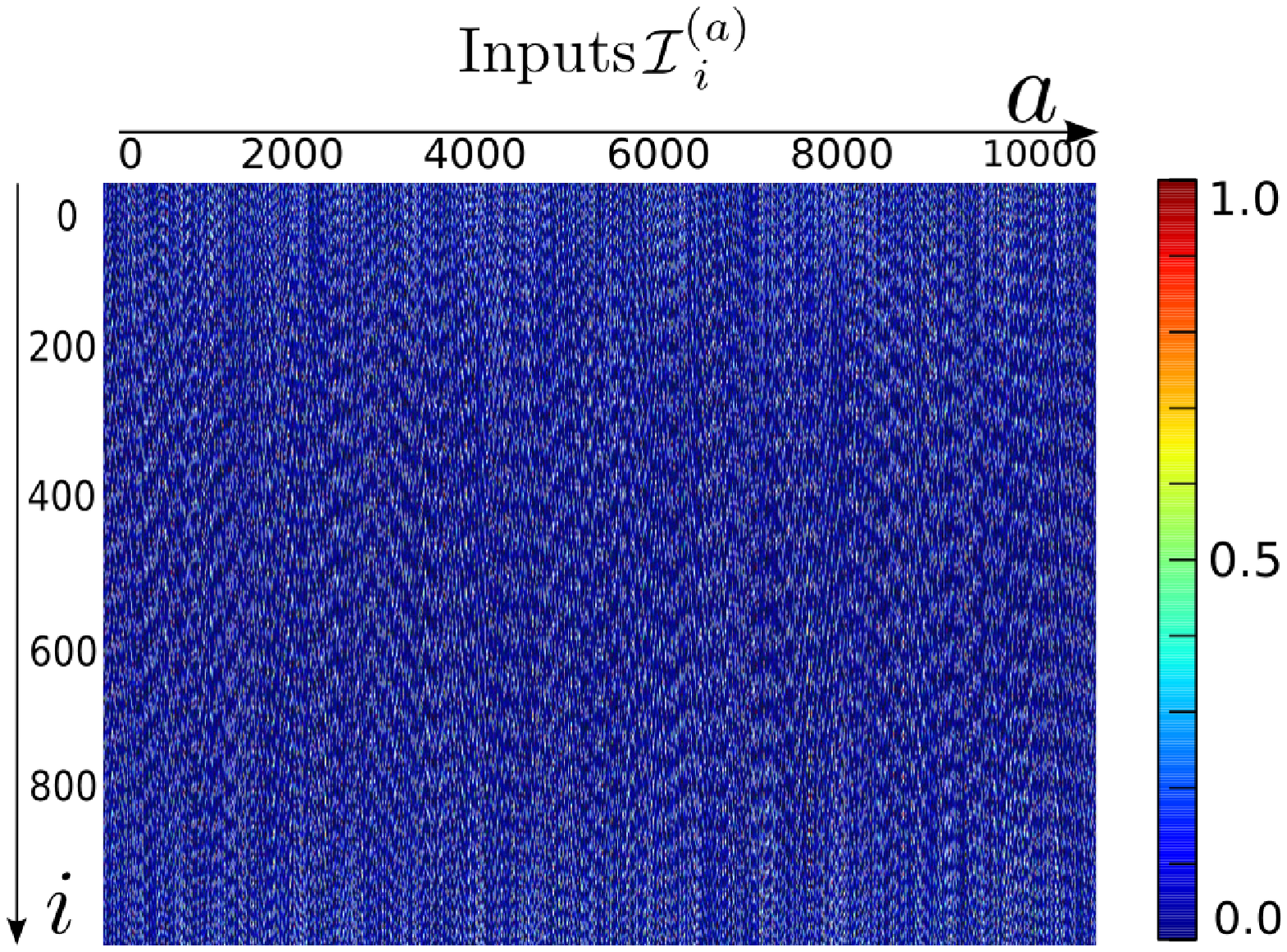}\hfill
 \includegraphics[width=0.45\textwidth]{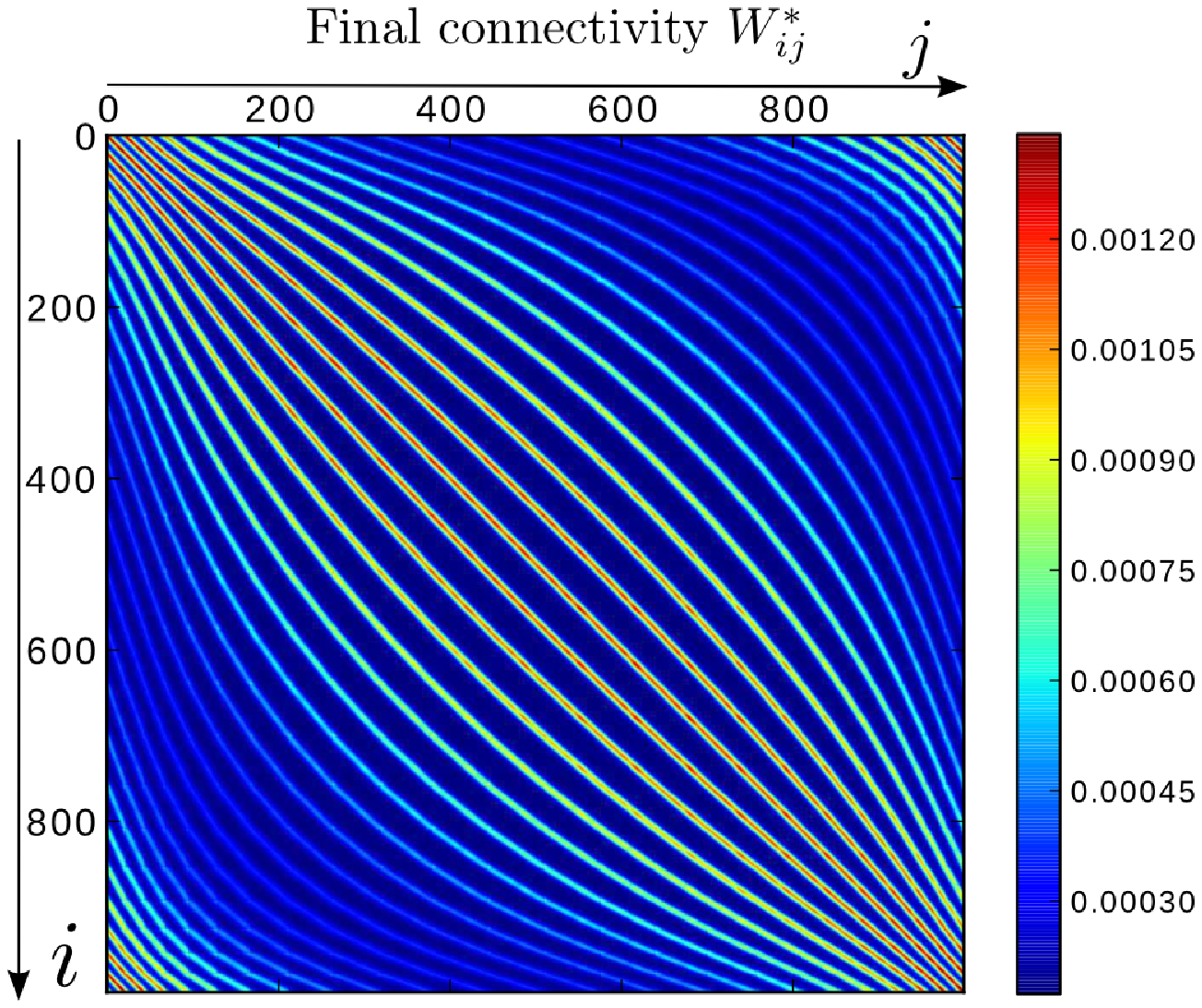}
 \caption{Plot of retinotopic inputs on $\Omega = \mathbb{T}^2$ (left) and the final connectivity matrix (right) for the system $\Sigma'$. The parameters used for this simulation are $s(x) = \frac{1}{1+e^{-4(x - 1)}}$, $l = 1$, $\mu = 10$, $\eps = 0.001$, $N=1000, M=10,000$, $\sigma = 2$.}
 \label{fig: torus simus}
\end{figure}

\begin{figure}[htbp]
 \centering
 \includegraphics[width=0.5\textwidth]{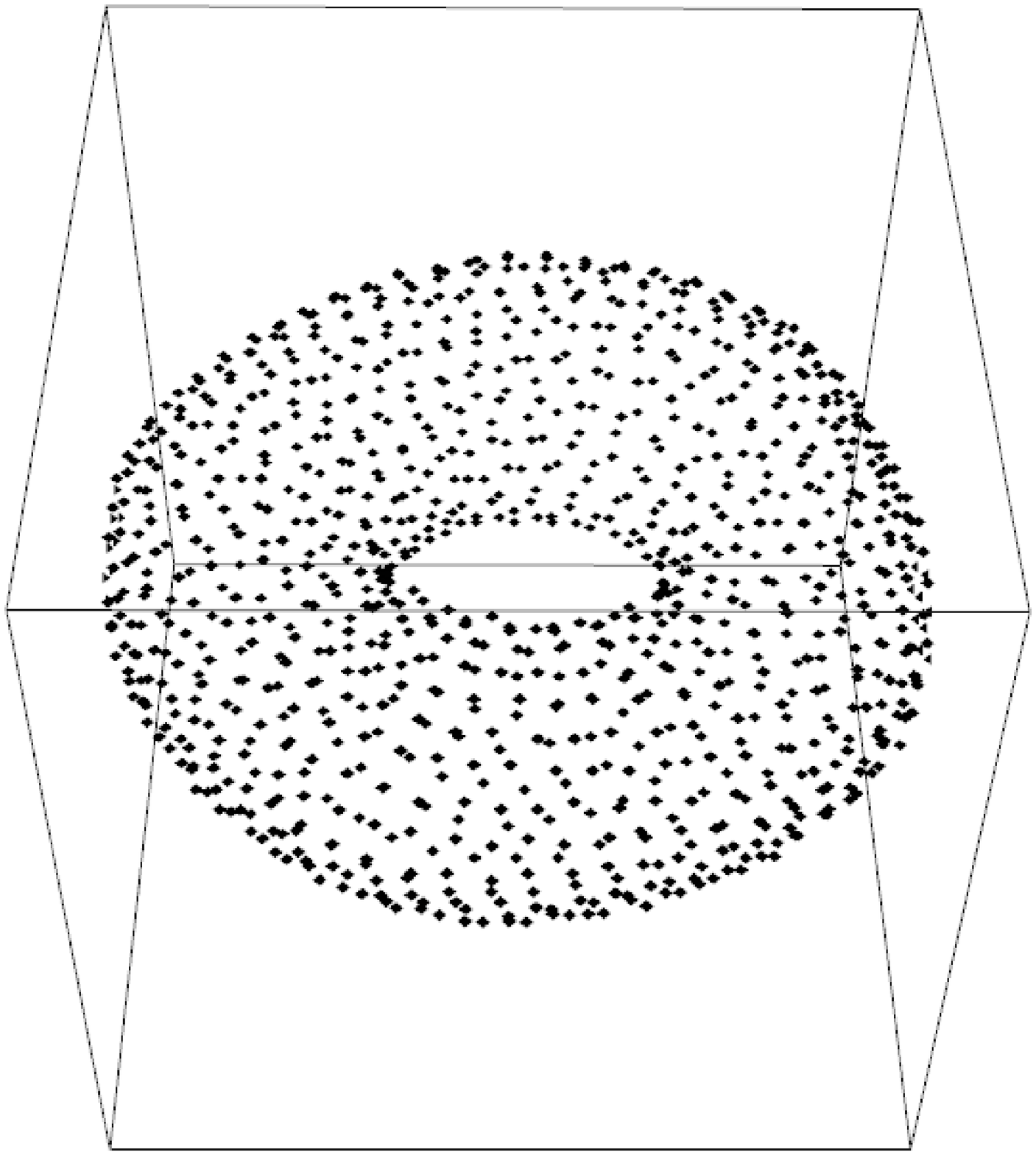}\hfill
\includegraphics[width=0.5\textwidth]{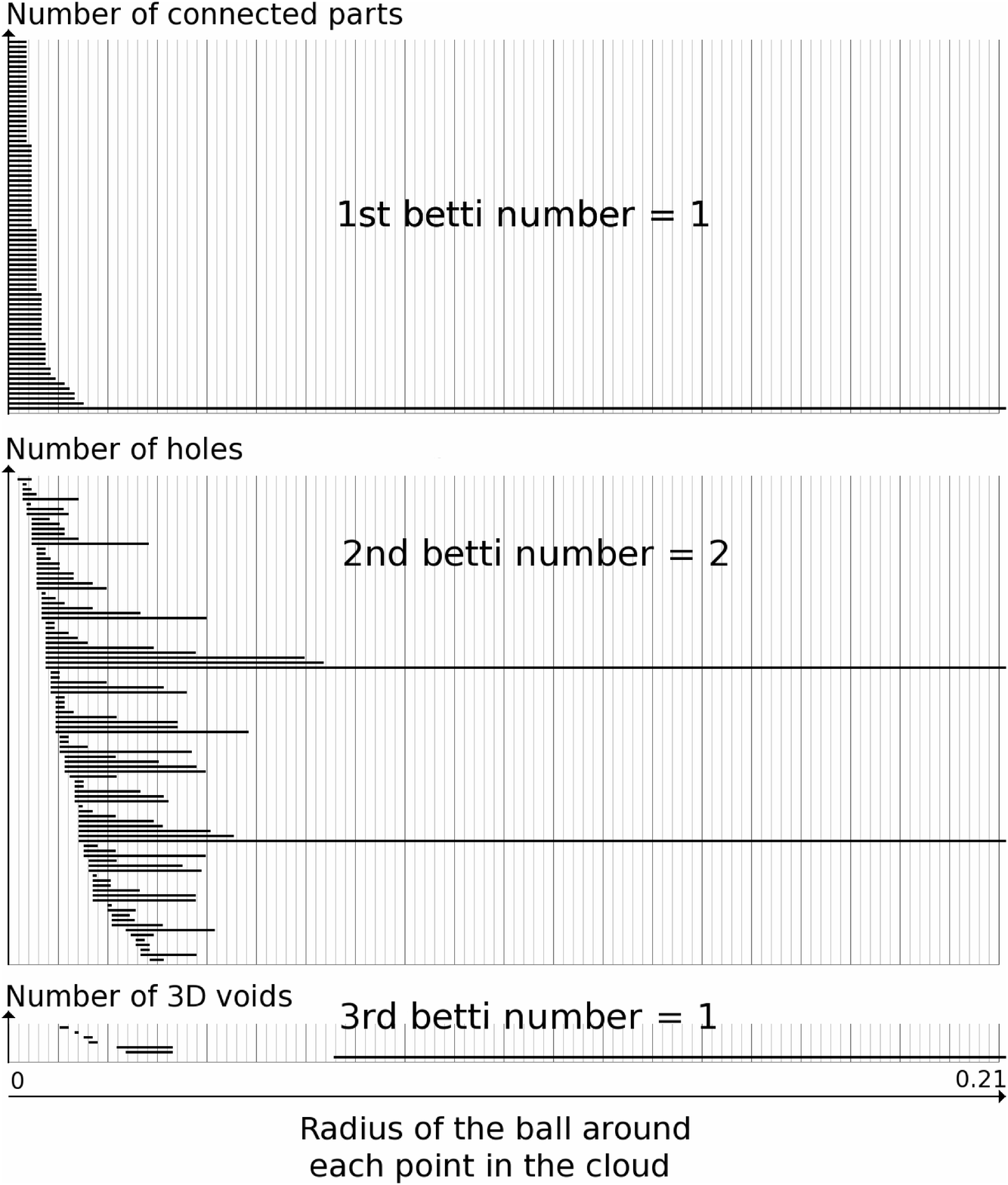}
 \caption{Left: Positions $x_i$ of the neurons for $k = 3$ after having applied multidimensional scaling methods presented in part \ref{part: multidimensional scaling} to the final connectivity matrix shown in figure \ref{fig: torus simus}. Right: Persistent cohomology barcodes of the cloud of points $x_i$ computed using the Jplex software package of \cite{plex}. (See section~\ref{part: persistent cohomology} for a short introduction to persistent cohomology). The triplet of betti numbers (1,2,1) appear stable confirming that the points lie on a 2 dimensional torus.}
 \label{fig: torus positions}
\end{figure}

We now assume that the inputs are uniformly distributed over a two-dimensional torus, i.e. $\Omega = \mathbb{T}^2$. That is, the input  labels $z_a$ are randomly distributed on the torus. The neuron labels $y_i$ are regularly and uniformly distributed on the torus. The inputs and final stable weight matrix of the simulated system are shown in figure~\ref{fig: torus simus}. The positions $x_i$ of the neurons after multidimensional scaling for $k=3$ are shown in figure~\ref{fig: torus positions}, and appear to form a cloud of points distributed on a torus. In order to confirm this, we have used a numerical method from computational cohomology \cite{zomorodian2005computing} to construct the so--called persistent cohomology barcodes of the neurons' positions. These determine certain topological invariants of the underlying space. (See section~\ref{part: persistent cohomology} for a short introduction to persistent cohomology and barcodes). The results are also shown in figure~\ref{fig: torus positions}, and establish that the network has learnt the underlying toroidal geometry of the inputs.

\subsection{Links with neuroanatomy}
The brain is subject to energy constraints which are completely neglected in the above formulation. These constraints most likely have a significant impact on the positions of real neurons in the brain. Indeed, it seems reasonable to assume that the positions and connections of neurons reflect a trade-off between the energy costs of biological tissue and their need to process information effectively. For instance, it has been suggested that a principle of wire length minimization may occur in the brain \cite{swindale1996development, chklovskii2002wiring}. In our neural mass framework, one may consider that the stronger two neural masses are connected, the larger the number of real axons linking the neurons together. Therefore, minimizing axonal length can be read as: the stronger the connection the closer, which is consistent with the convolutional part of the weight matrix. However, the underlying geometry of natural inputs is likely to be very high-dimensional, whereas the brain lies in a three-dimensional world. In fact, the cortex is so flat that it is effectively two-dimensional. Hence, the positions of real neurons are different from the positions $x_i \in \R^k$ in a high dimensional vector space; since the cortex is roughly two-dimensional, the positions could only be realized physically if $k = 2$. Therefore, the three-dimensional toric geometry or any higher dimensional structure could not be perfectly implemented in the cortex without the help of non-convolutional long-range connections. Indeed, we suggest that the cortical connectivity is made of two parts: i) a local convolutional connectivity corresponding to the convolutional term $G_\sigma$ in (\ref{eq: connectivity decomposition}), which is consistent with the requirements of energy efficiency, and ii) a non-convolutional connectivity corresponding to the factor $M$ in equation (\ref{eq: connectivity decomposition}), which is required in order to represent various stimulus features. If the cortex were higher-dimensional ($k \gg 2$) then $M\equiv 1$.

We illustrate the above claim by considering two examples based on the functional anatomy of the primary visual cortex: the emergence of ocular dominance columns and orientation columns, respectively. We  proceed by returning to the case of planar retinotopy (section 5.3.1) but now with additional input structure. In the first case, the inputs are taken to be binocular and isotropic, whereas in the second case they are taken to be monocular and anisotropic. The details are presented below. Given a set of prescribed inputs, the network evolves according to equation (\ref{eq: system Sigma prime}) and the lateral connections converge to a stable equilibrium. The resulting weight matrix is then projected onto the set of distance matrices for $k=2$ (as described in section \ref{part: multidimensional scaling}) using the stress majorization or SMACOF algorithm for the stress1 cost function as described in \cite{borg2005modern}. We thus assign a position $x_i\in \R^2$ to the $i$th neuron, $i=1,\ldots, N$. (Note that the position $x_i$ extracted from the weights using multidimensional scaling is distinct from the ``physical'' position $y_i$ of the neuron in the retinocortical plane; the latter determines the center of its receptive field). The convolutional connectivity ($G_\sigma$ in equation \ref{eq: connectivity decomposition}) is therefore completely defined: on the planar map of points $x_i$, neurons are isotropically connected to their neighbors; the closer the neurons are the stronger is their convolutional connection. Moreover, since the stimulus feature preferences (orientation, ocular dominance) of each neuron $i$, $i=1,\ldots, N$, are prescribed, we can superimpose these feature preferences on to the planar map of points $x_i$. In both examples, we find that neurons with the same ocular or orientation selectivity tend to cluster together (see figures \ref{fig: ocular dominance} and \ref{fig: orientation}): interpolating these clusters then generates corresponding feature columns. It is important to emphasize that the retinocortical positions $y_i$ do not have any columnar structure, that is, they do not form clusters with similar feature preferences. Thus, in contrast to standard developmental models of vertical connections, the columnar structure emerges from the recurrent weights following Hebbian learning and an application of multidimensional scaling. It follows that neurons coding for the same feature tend to be strongly connected; indeed, the multidimensional scaling algorithm has the property that it positions  strongly connected neurons close together . Equation (\ref{eq: connectivity decomposition}) also suggests that the connectivity has a non-convolutional part,  $M$, which is a consequence of the low-dimensionality ($k=2$). In order to illustrate the structure of the non-convolutional connectivity, we select a neuron $i$ in the plane and draw a link from it at position $x_i$ to the neurons at position $x_j$ for which $M(x_i,x_j)$ is maximal. We find that $M$ tends to be patchy, i.e. it connects neurons having the same feature preferences. In the case of orientation, $M$ also tends to be co-aligned, i.e. connecting neurons with similar orientation preference along a vector in the plane of the same orientation.

\subsubsection{Ocular dominance columns and patchy connectivity}
In order to construct binocular inputs, we partition the $N$ neurons into two sets $i\in \{1,\ldots,N/2\}$ and $i\in \{N/2+1,\ldots,N\}$ that code for the left and right eyes, respectively. The $i$th neuron is then given a retinocortical position $y_i\in [0,1]\times [0,1]$, with the $y_i$ uniformly distributed across the plane. We do not assume {\em a priori} that there exist any ocular dominance columns, that is, neurons with similar retinocortical positions $y_i$ do not form clusters of cells coding for the same eye. We then take the $a$th input to the network to be of the form
\begin{align*}
I_i^{(a)} &= (1+\gamma(a)) e^{-\frac{(y_i-z_a)^2}{\sigma'^2}},\quad i=1,\ldots,N/2 \\
I_i^{(a)} &= (1-\gamma(a)) e^{-\frac{(y_i-z_a+s)^2}{\sigma'^2}},\quad i=N/2+1,\ldots,N,
\end{align*}
where $s \in \R^2$ represents some form of binocular disparity, $z_a$ and $\gamma(a)$ are randomly generated from $[0,1]^2$ and $[-1,1]$, respectively, see \cite{bressloff:05}. Thus, if $\gamma(a) > 0$ ($\gamma(a) < 0)$ then the corresponding input is predominantly from the left (right) eye. In our simulations we take $\sigma = \sigma' = 0.1$ and $s = 0.005$. The results of our simulations are shown in figure ~\ref{fig: ocular dominance}. In particular, we plot the points $x_i$ obtained by performing multidimensional scaling on the final connectivity matrix for $k=2$, and superimposing upon this the ocular dominance map obtained by interpolating between clusters of neurons with the same eye preference. We also illustrate the non-convolutional connectivity by linking one selected neuron to the five neurons labeled $j$ it is most strongly connected to (with $M(x_{i}, x_j)>1$), with $i$ the label of the central neuron. This clearly shows that long--range connections tend to link cells with the same ocular dominance.

\begin{figure}[htbp]
 \centering
 \includegraphics[width=0.9\textwidth]{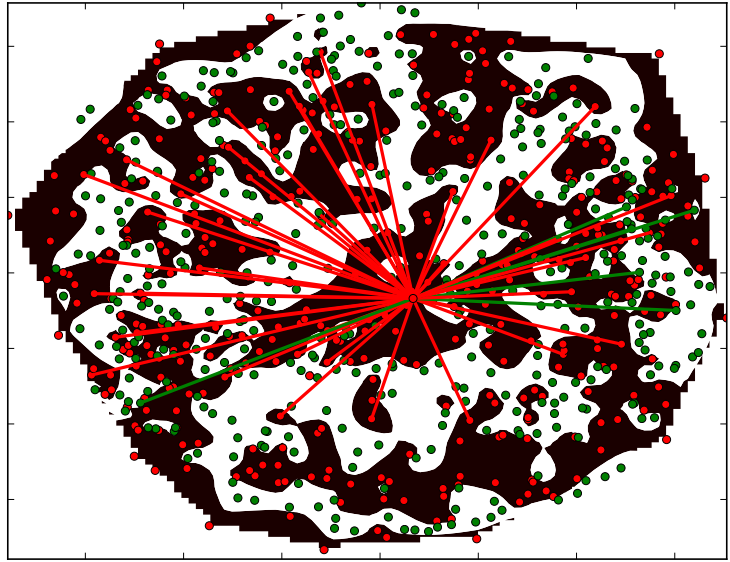}
 \caption{Plot of the positions $x_i$ of neurons for $k=2$ in red (right eye) and green (left eye). We have used an interpolation method to highlight the areas dominated by the right eye in black. These ocular dominance columns have fractal borders which are less regular than those observed in optical imaging experiments. The convolutional connectivity ($G_\sigma$ in equation (\ref{eq: connectivity decomposition})) is implicitly described by the position of the neurons: the closer the neurons, the stronger their connections. The strongest components of the non-convolutional connectivity ($M$ in equation (\ref{eq: connectivity decomposition})) from a central red neuron are also shown by drawing links from this neuron to the target neurons. The color of the link refers to the color of the target neuron. Therefore, we see that it is mainly connected to neurons of its same ocular dominance resulting in a patchy distribution. The parameters used for this simulation are $s(x) = \frac{1}{1+e^{-4(x - 1)}}$, $l = 1$, $\mu = 10$, $\eps = 0.01$, $N=800$ $M = 3200$. }
 \label{fig: ocular dominance}
\end{figure}

\subsubsection{Orientation columns and colinear connectivity}

In order to construct oriented inputs, we partition the $N$ neurons into four groups $\Sigma_{\theta}$ corresponding to different orientation preferences $\theta = \{0, \frac{\pi}{4}, \frac{\pi}{2},\frac{3\pi}{4}\}$. Thus, if neuron $i \in \Sigma_{\theta}$ then its orientation preference is $\theta_i=\theta$. For each group, the neurons are randomly assigned a retinocortical position $y_i \in [0,1]\times [0,1]$. Again, we do not assume {\em a priori} that there exist any orientation columns, that is, neurons with similar retinocortical positions $y_i$ do not form clusters of cells coding for the same orientation preference. Each cortical input ${I}_i^{(a)}$ is generated by convolving a thalamic input consisting of an oriented Gaussian with a Gabor--like receptive field \cite{miikkulainen-bednar-etal:05}. Let $\mathcal{R}_\theta$ denote a 2-dimensional rigid body rotation in the plane with $\theta \in [0,2\pi)$. Then
\begin{equation}
 I_i^{(a)}=\int G_i(\xi-y_i)I_a(\xi-z_a)d\xi,
\end{equation}
where
\begin{equation}
G_i(\xi)=G_0(\mathcal{R}_{\theta_i} \xi )
\end{equation}
and $G_0(\xi)$ is the Gabor--like function
 $$G_0(\xi) = A_+ e^{-\xi^T.\Lambda^{-1}.\xi} - A_- e^{-(\xi - e_0)^T.\Lambda^{-1}.(\xi - e_0)} - A_- e^{-(\xi + e_0)^T.\Lambda^{-1}.(\xi + e_0)}$$
 with $e_0 = (0,1)$ and
$$\Lambda= \begin{pmatrix} \sigma_{large} & 0 \\ 0 & \sigma_{small} \end{pmatrix}.$$
The amplitudes $A_+, A_-$ are chosen so that $\int G_0(\xi)d\xi = 0$.
Similarly, the thalamic input $I_a(\xi)=I(\mathcal{R}_{\theta'_a}\xi)$ with $I(\xi)$ the anisotropic Gaussian
\[I(\xi) = e^{-\xi^T.\Lambda'^{-1}.\xi}, \qquad \Lambda' = \begin{pmatrix} \sigma'_{large} & 0 \\ 0 & \sigma'_{small} \end{pmatrix}.\]
The input parameters $\theta'_a$ and $z_a$ are randomly generated from $[0,\pi)$ and $[0,1]^2$ respectively. In our simulations we take $\sigma_{large} = 0.133...$, $\sigma'_{large} = 0.266...$ and $\sigma_{small} = \sigma'_{small} = 0.0333...$. The results of our simulations are shown in figure ~\ref{fig: orientation}. In particular, we plot the points $x_i$ obtained by performing multidimensional scaling on the final connectivity matrix for $k=2$, and superimposing upon this the orientation preference map obtained by interpolating between clusters of neurons with the same orientation preference. To avoid border problems we have zoomed on the center on the map. We also illustrate the non-convolutional connectivity by linking one selected neuron to all other neurons for which $M$ is maximal. The patchy, anisotropic nature of the long--range connections is clearly seen. The anisotropic nature of the connections is further quantified in the histogram of figure~\ref{fig: orientation histogram}. 

\begin{figure}[htbp]
 \centering
 \includegraphics[width=0.8\textwidth]{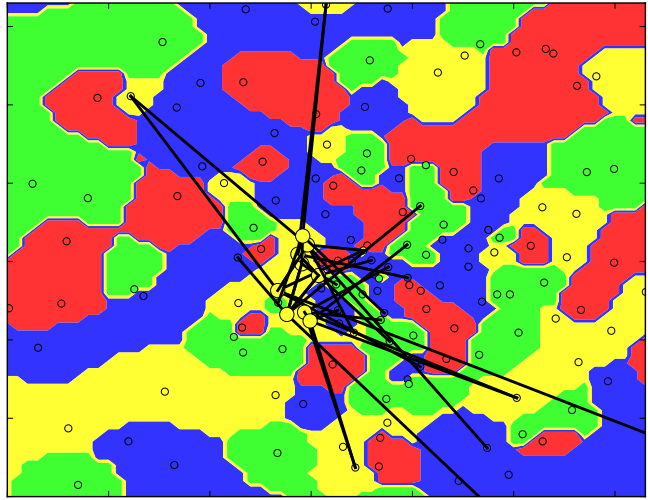}\hfill
 \caption{Plot of the positions $x_i$ of neurons for $k=2$ obtained by multidimensional scaling of the weight matrix. Neurons are clustered in orientation columns represented by the colored areas, which are computed by interpolation. The strongest components of the non-convolutional connectivity ($M$ in equation \eqref{eq: connectivity decomposition}) from a particular neuron in a yellow area are illustrated by drawing black links from this neuron to the target neurons. Since the yellow color corresponds to an orientation of $\frac{3\pi}{4}$, the non-convolutional connectivity shows the existence of a co-linear connectivity as exposed in \cite{bosking1997orientation}. The parameters used for this simulation are $s(x) = \frac{1}{1+e^{-4(x - 1)}}$, $l = 1$, $\mu = 10$, $\eps = 0.01$, $N=900$ $M = 9000$.}
 \label{fig: orientation}
\end{figure}

\begin{figure}[htbp]
 \centering
 \includegraphics[width=0.8\textwidth]{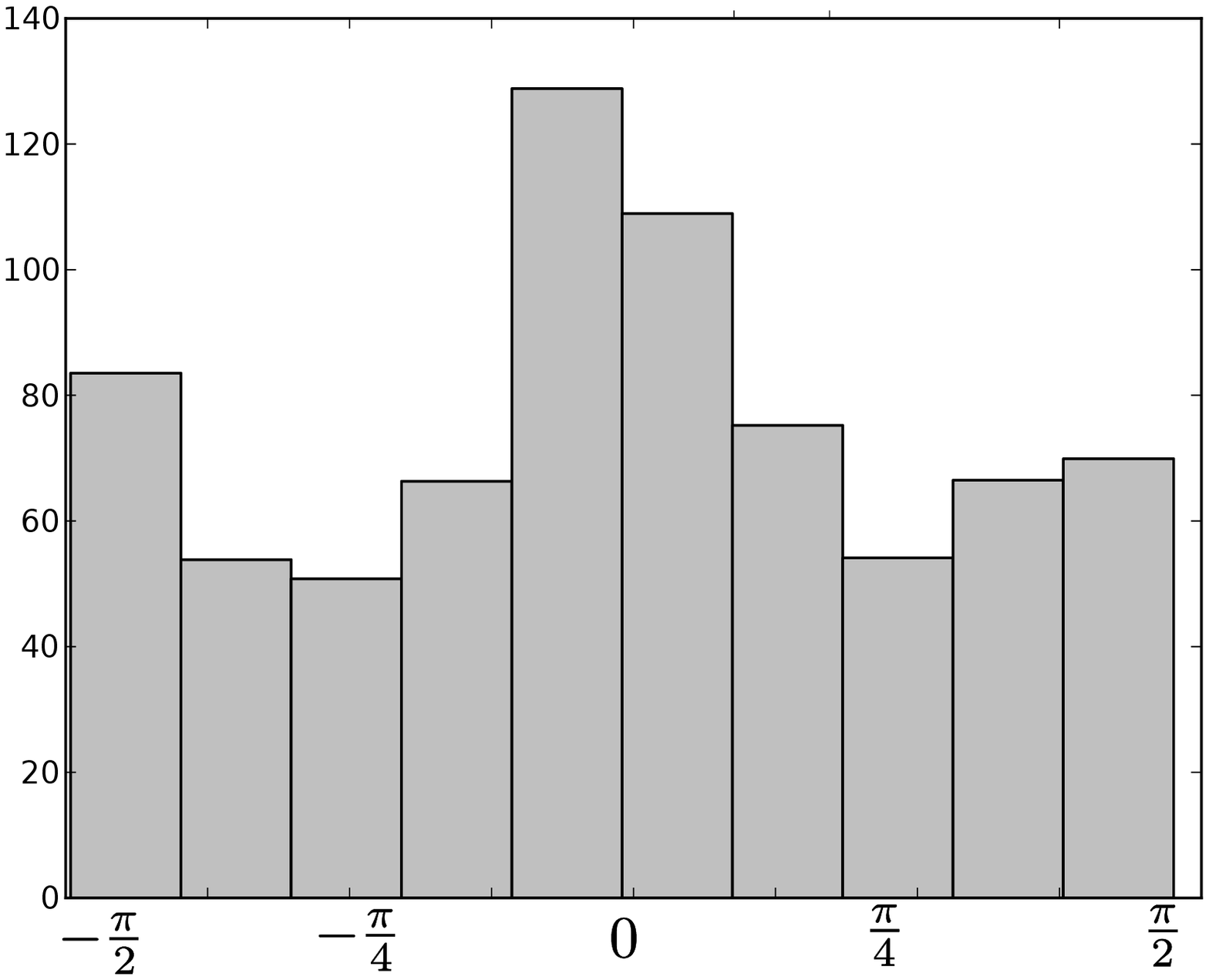}
 \caption{Histogram of the 5 largest components of the non-convolutional connectivity for 80 neurons randomly chosen among those shown in Fig. \ref{fig: orientation}. The abcissa corresponds to the difference in radian between the direction preference of the neuron and the direction of the links between the neuron and the target neurons. This histogram is weighted by the strength of the non-convolutional connectivity. It shows a preference for co-aligned neurons but also a slight preference for perpendicularly-aligned neurons (e.g. neurons of the same orientation but parallel to each other).}
 \label{fig: orientation histogram}
\end{figure}
\bigskip

\section{Discussion}
\label{part: discussion}
In this paper, we have shown how a neural network can learn the underlying geometry of a set of inputs. We have considered a fully recurrent neural network whose dynamics is described by a simple non-linear rate equation, together with unsupervised Hebbian learning with decay that occurs on a much slower time scale. Although several inputs are periodically presented to the network, so that the resulting dynamical system is non-autonomous, we have shown that such a system has a fairly simple dynamics: the network connectivity matrix always converges to an equilibrium point. We have then demonstrated how this connectivity matrix can be expressed as a distance matrix in $\R^k$ for sufficiently large $k$, which can be related to the underlying geometrical structure of the inputs. If the connectivity matrix is embedded in a lower two-dimensional space ($k=2$), then the emerging patterns are similar to experimentally observed cortical feature maps. That is, neurons with the same feature preferences tend to cluster together forming cortical columns within the embedding space. Moreover, the recurrent weights decompose into a local isotropic convolutional part, which is consistent with the requirements of energy efficiency, and a longer--range non-convolutional part that is patchy. This suggest a new interpretation of the cortical maps: they correspond to two-dimensional embeddings of the underlying geometry of the inputs.

One of the limitations of applying simple Hebbian learning to recurrent cortical connections is that it only takes into account excitatory connections, whereas 20$\%$ of cortical neurons are inhibitory. Indeed, in most developmental models of feedforward connections, it is assumed that the local and convolutional connections in cortex have a Mexican hat shape with negative (inhibitory) lobes for neurons that are sufficiently far from each other. From a computational perspective, it is possible to obtain such a weight distribution by replacing Hebbian learning with some form of covariance learning  (\cite{sejnowski1989hebb}). However, it is difficult to prove convergence to a fixed point in the case of the covariance learning rule, and multidimensional scaling method cannot be applied directly unless the Mexican hat function is truncated so that it is invertible. Another limitation of rate-based Hebbian learning is that it does not take into account causality, in contrast to more biologically detailed mechanisms such as spike timing dependent plasticity.

The approach taken here is very different from standard treatments of cortical development \cite{miller89,swindale1996development}, in which the recurrent connections are assumed to be fixed and of convolutional Mexican hat form whilst the feedforward vertical connections undergo some form of correlation-based Hebbian learning. In the latter case, cortical feature maps form in the physical space of retinocortoical coordinates $y_i$, rather than in the more abstract planar space of points $x_i$ obtained by applying multidimensional scaling to recurrent weights undergoing Hebbian learning in the presence of fixed vertical connections. A particular feature of cortical maps formed by modifiable feedforward connections is that the mean size of a column is determined by a Turing-like pattern forming instability, and depends on the length scales of the Mexican hat weight function and the two-point input correlations \cite{miller89,swindale1996development}. No such Turing mechanism exists in our approach so that the resulting cortical maps tend to be more fractal-like (many length scales) compared to real cortical maps. Nevertheless, we have established that the geometrical structure of cortical feature maps can also be encoded by modifiable recurrent connections. This should have interesting consequences for models that consider the joint development of feedforward and recurrent cortical connections. One possibility is that the embedding space of points $x_i$ arising from multidimensional scaling of the weights becomes identified with the physical space of retinocortical positions $y_i$. The emergence of local convolutional structures together with sparser long-range connections would then be consistent with energy efficiency constraints in physical space.

Our paper also draws a direct link between the recurrent connectivity of the network and the positions of neurons in some vector space such as $\R^2$. In other words, learning corresponds to moving neurons or nodes so that their final position will match the inputs' geometrical structure. Similarly, the Kohonen algorithm \cite{kohonen1990self} describes a way to move nodes according to the inputs presented to the network. It also converges toward the underlying geometry of the set of inputs. Although not formally equivalent, it seems that both of these approaches have the same qualitative behaviour. However, our method is more general in the sense that no neighborhood structure is assumed {\em a priori}; such a  structure emerges via the embedding into $\R^k$.

Finally, note that we have used a discrete formalism based on a finite number of neurons. However, the resulting convolutional structure obtained by expressing the weight matrix as a distance matrix in $\R^k$, equation (\ref{eq: connectivity decomposition}), allows us to take an appropriate continuum limit. This then generates a continuous neural field model in the form of an integro-differential equation whose integral kernel is given by the underlying weight distribution. Neural fields have been used increasingly to study large--scale cortical dynamics (see \cite{coombes2005waves} for a review). Our geometrical learning theory provides a developmental mechanism for the formation of these neural fields. One of the useful features of neural fields from a mathematical perspective is that many of the methods of partial differential equations can be carried over. Indeed, for a general class of connectivity functions defined over continuous neural fields, a reaction-diffusion equation can be derived whose solution approximates the firing rate of the associated neural field \cite{degond1989weighted,cottet1995neural,edwards1996approximation}. It appears that the necessary connectivity functions are precisely those that can be written in the form (\ref{eq: connectivity decomposition}). This suggests that a network that has been trained on a set inputs with an appropriate geometrical structure behaves as a diffusion equation in a high-dimensional space together with a reaction term corresponding to the inputs.

\section{Acknowldegments}
MNG and ODF were partially funded by the ERC advanced grant NerVi nb 227747. MG was partially funded by the r\'egion PACA, France. This publication was based on work supported in part by the National Science Foundation (DMS-0813677) and by Award No KUK-C1-013-4 made by King Abdullah University of Science and Technology (KAUST). PCB was also partially supported by the Royal Society Wolfson Foundation.

\section{Appendix}
\subsection{Proof of the convergence to the symmetric attractor $\mathcal{A}$}
\label{part: appendix symmetric fixed points}
We need to prove the 2 points: (i) $\mathcal{A}$ is an invariant set, and (ii) for all ${\mathcal Y}(0) \in \R^{N \times M} \times \R^{N \times N}$, ${\mathcal Y}(t)$ converges to $\mathcal{A}$ as $t \rightarrow +\infty $. Since $\R^{N \times N}$ is the direct sum of the set of symmetric connectivities and the set of anti-symmetric connectivies, we write $W(t) =  W_S(t) + W_A(t),\ \forall t \in \R_+$, where $W_S$ is symetric and $W_A$ is anti-symetric.

(i) In (\ref{eq: system Sigma prime}),  the right hand side of the equation for $\dot{W}$ is symmetric. Therefore,  if $\exists t_1 \in R_+$ such that $W_A(t_1) = 0$, then W remains in $\mathcal{A}$ for $t\geq t_1$.

(ii) Projecting the expression for $\dot{W}$ in equation (\ref{eq: system Sigma prime}) on to the anti-symmetric component leads to
\begin{equation}
 \frac{dW_A}{dt} = -\eps \mu W_A(t)
\end{equation}
whose solution is $W_A(t) = W_A(0) \exp(-\eps \mu t), \forall t \in \R_+$. Therefore, $\displaystyle \lim_{t \to +\infty} W_A(t) = 0$. The system converges exponentially to $\mathcal{A}$.

\subsection{Proof of theorem \ref{thm: lasalle implies convergence}}
\label{part: appendix Lasalle}
Consider the following Lyapunov function (see equation (\ref{eq: energy learning}))
\begin{equation}
E({\mathcal U},W) = -\frac{1}{2}\langle {\mathcal U},W\cdot {\mathcal U} \rangle - \langle {\mathcal I},{\mathcal U} \rangle +{\langle 1, \overline{S^{-1}}\big({\mathcal U}\big) \rangle} + \frac{\tilde{\mu}}{2} \|W\|^2,
\end{equation}
where $\tilde{\mu}=\mu M$, such that if $W= W_S + W_A$, where $W_S$ is symmetric and $W_A$ is anti-symmetric.
\begin{equation}
 -\nabla E({\mathcal U},W) = \begin{pmatrix}
		W_S\cdot {\mathcal U} + I - S^{-1}\big({\mathcal U}\big) \\
		{\mathcal U} \cdot {\mathcal U}^T - \mu W
           \end{pmatrix}
\end{equation}
Therefore, writing the system $\Sigma'$, equation (\ref{eq: system Sigma prime}), as
$$\frac{d {\mathcal Y}}{dt}= \gamma \begin{pmatrix}
		W_S \cdot S({\mathcal V}) + I - S^{-1}\big(S({\mathcal V})\big) \\
		S({\mathcal V}) \cdot  S({\mathcal V})^T - \tilde{\mu} W
           \end{pmatrix} +\gamma \begin{pmatrix}
		W_A.S({\mathcal V})\\0\end{pmatrix},$$
where ${\mathcal Y}=( {\mathcal V},W)^T$, we see that
\begin{equation}
\label{why}
\frac{d{\mathcal Y}}{dt} = -\gamma\bigg(\nabla E \big(\sigma({\mathcal V}, W)\big)\bigg) + \Gamma(t)
\end{equation}
where $\gamma({\mathcal V},W)^T = ( {\mathcal V},\eps W/ M )^T$, $\sigma({\mathcal V},W) = ( S({\mathcal V}), W)$ and $\Gamma: \R_+ \rightarrow \H$ such that $\|\Gamma\| \underset{t\rightarrow +\infty}{\rightarrow} 0$ exponentially (because the system converges to $\mathcal{A}$).
It follows that the time derivative of $\tilde{E} = E \circ \sigma$ along trajectories is given by:
\begin{equation}
 \frac{d\tilde{E}}{dt} = \bigg\langle \nabla \tilde{E} , \frac{d{\mathcal Y}}{dt}\bigg\rangle = \bigg\langle \nabla_{\mathcal V} \tilde{E} , \frac{d{\mathcal V}}{dt}\bigg\rangle+\bigg\langle \nabla_W \tilde{E} , \frac{dW}{dt}\bigg\rangle .
\end{equation}
Substituting equation (\ref{why}) then yields
\begin{eqnarray}
\frac{d\tilde{E}}{dt} &=&  - \bigg\langle  \nabla \tilde{E}, \gamma\big(\nabla E \circ \sigma \big) \bigg\rangle + \underbrace{\bigg\langle \nabla \tilde{E}, \Gamma(t) \bigg\rangle}_{\tilde{\Gamma}(t)} \\
&=& - \bigg\langle S'({\mathcal V})\nabla_{\mathcal U} E \circ \sigma, \nabla_{\mathcal U}  E \circ \sigma \bigg\rangle - \frac{\eps}{M} \bigg\langle \nabla_W E \circ \sigma, \nabla_W E \circ \sigma \bigg\rangle + \tilde{\Gamma}(t) .\nonumber 
\end{eqnarray}
We have used the chain--rule of differentiation, whereby 
\[\nabla_V(\tilde{E})=\nabla_V(E\circ \sigma)= S'({\mathcal V})\nabla_{\mathcal U} E \circ \sigma , \]
and $S'({\mathcal V})\nabla_{\mathcal U} E $ (without dots) denotes the Hadamard (term by term) product, that is, 
\[ [S'({\mathcal V})\nabla_{\mathcal U} E]_{ia}=s'({ V}_i^{(a)})\frac{\partial E}{\partial { U}_i^{(a)}}\]

Note that $|\tilde{\Gamma}| \underset{t\rightarrow +\infty}{\rightarrow} 0$ exponentially because $\nabla \tilde{E}$ is bounded, and $S'({\mathcal V}) >0$ because the trajectories are bounded. Thus, there exists $t_1 \in \R_+$ such that $\forall t > t_1$, $\exists k \in \R_+^*$ such that
\begin{equation}
\frac{d\tilde{E}}{dt} \leq -k \|\nabla E \circ \sigma \|^2 \leq 0 .
\label{eq: dE/dt<0}
\end{equation}
As in \cite{cohen-grossberg:83} and \cite{dong1992dynamic}, we apply the Krasovskii-LaSalle invariance principle \cite{khalil-grizzle:96}. We check that:
\begin{itemize}
 \item $\tilde{E}$ is lower bounded. Indeed, ${\mathcal V}$ and $W$ are bounded. Given that ${\mathcal I}$ and $S$ are also bounded it is clear that $\tilde{E}$ is bounded.
\item $\displaystyle \frac{d\tilde{E}}{dt}$ is negative semidefinite on the trajectories as shown in equation (\ref{eq: dE/dt<0}).
\end{itemize}
Then the invariance principle tells us that the solutions of the system $\Sigma'$ approach the set $M = \Big\{ {\mathcal Y} \in \H : \displaystyle\frac{d\tilde{E}}{dt}({\mathcal Y})=0\Big\}$. Equation (\ref{eq: dE/dt<0}) implies that $M= \Big\{ Y \in \H : \nabla E \circ \sigma=0\Big\}$. Since $\displaystyle\frac{d{\mathcal Y}}{dt} = - \gamma\Big(\nabla E \circ \sigma\Big)$ and $\gamma \neq 0$ everywhere, $M$ consists of the equilibrium points of the system. This completes the proof.

\subsection{Proof of theorem \ref{thm: fixed points and linear stability}}
\label{part: appendix linear stability}
Denote the right--hand side of system $\Sigma'$, equation (\ref{eq: system Sigma prime}) by
\[F({\mathcal V},W) = \left\{
    \begin{array}{c}
	-{\mathcal V} +  W \cdot S\big({\mathcal V}\big)+I\\
        \displaystyle\frac{\eps}{M} \big(S({\mathcal V}).S({\mathcal V})^T  -\mu M W\big)
    \end{array}
\right.\]
The fixed points satisfy the condition $F({\mathcal V},W) = 0$ which immediately leads to equations (\ref{eq: M0 fixed points subspace}). Let us now check the linear stability of this system. The differential of $F$ at ${\mathcal V}^*,W^*$ is
\begin{equation}
 dF_{({\mathcal V}^*,W^*)} (\Z,J) = \begin{pmatrix}
               -\Z+W^*\cdot\big(S'({\mathcal V}^*)\Z\big) + J\cdot S({\mathcal V}^*)\\
		\displaystyle\frac{\eps}{M} \Big(\big(S'({\mathcal V}^*)\Z\big) \cdot S({\mathcal V}^*)^T + S({\mathcal V}^*) \cdot \big(S'({\mathcal V}^*)\Z\big)^T - \mu MJ \Big),
              \end{pmatrix} \nonumber
\label{eq: jacobian M0}
\end{equation}
where $S'({\mathcal V}^*)\Z$ denotes a Hadamard product, that is, $[S'({\mathcal V}^*)\Z]_{ia}= s'({V^*_i}^{(a)})Z_i^{(a)}$. Assume that there exist $\lambda\in \mathbb{C}^*,\ (\Z,J) \in \H$ such that $dF_{(V^*,W^*)} \begin{pmatrix}\Z\\J\end{pmatrix} = \lambda\begin{pmatrix}\Z\\J\end{pmatrix}$. Taking the second component of this equation and computing the dot product with $S({\mathcal V}^*)$ leads to
$$({\lambda + \eps \mu}) J\cdot S = \frac{\eps}{M}\left ((S'\Z) \cdot  S^T\cdot S + S\cdot  (S' \Z)^T\cdot S \right )$$
where $S =S({\mathcal V}^*)$, $S'=S'({\mathcal V}^*)$. Substituting this expression in the first equation leads to
\begin{multline}
M(\lambda + \eps \mu)(\lambda + 1) \Z = (\frac{\lambda}{\mu} + \eps) S \cdot  S^T\cdot (S'\Z) 
+ {\eps} (S'\Z)  \cdot  S^T\cdot  S
+ \eps S \cdot (S'\Z)^T \cdot S
\end{multline}

Observe that setting $\eps=0$ in the previous equation leads to an eigenvalue equation for the membrane potential only: 
\[ (\lambda + 1) \Z =\frac{1}{\mu M} S \cdot S^T\cdot (S'\Z).\]
Since $W^* = \frac{1}{\mu M} \big(S \cdot S^T \big)$, this equation implies that $\lambda + 1$ is an eigenvalue of the operator $X \mapsto W^*.(S'X)$. The magnitudes of the eigenvalues are always smaller than the norm of the operator. Therefore, we can say that if $1>\|W^*\| S'_m$ then all the possible eigenvalues $\lambda$ must have a negative real part. This sufficient condition for stability is the same as in \cite{faugeras-grimbert-etal:08}. It says that fixed points sufficiently close to the origin are always stable.

Let us now consider the case $\eps \neq 0$. Recall that $\Z$ is a matrix. We now ``flatten'' $\Z$ by storing its rows in a vector called $\Z_{row}$. We use the following result in \cite{brewer:78}: the matrix notation of operator $X \mapsto A\cdot X \cdot B$ is $A\otimes B^T$, where $\otimes$ is the Kronecker product. In this formalism the previous equation becomes
\begin{multline}
M(\lambda + \eps \mu)(\lambda + l) \Z_{row} = \bigg((\frac{\lambda}{\mu}+\eps)  S \cdot S^T \otimes I_d
+ \eps I_d \otimes S^T \cdot S
+ \eps S \otimes S^T\bigg) \cdot(S'Z)_{\text{row}}
\label{eq: necessary condition for eigenvalues of M0}
\end{multline}
where we assume that the Kronecker product has the priority over the dot product. We focus on the linear operator $\O$ defined  by the right hand side and bound its norm. Note that we use the following norm $\|W\|_\infty = \sup_X\frac{\|W.X\|}{\|X\|}$ which is equal to the largest magnitude of the eigenvalues of $W$.
\begin{multline}
 \|\O\|_\infty \leq S'_m \bigg( |\frac{\lambda}{\mu}|  \|S \cdot S^T \otimes I_d\|_\infty
+ \eps \|S \cdot S^T \otimes I_d\|_\infty + \eps \|I_d \otimes S^T \cdot S\|_\infty\\ + \eps \|S\otimes S^T\|_\infty \bigg).
\end{multline}
Define, $\nu_m$ to be the magnitude of the largest eigenvalue of $W^* = \frac{1}{\mu M}(S \cdot S^T)$.
First, note that $S\cdot S^T$ and $S^T \cdot S$ have the same eigenvalues $(\mu M) \nu_i$ but different eigenvectors denoted by $u_i$ for $S \cdot S^T$ and $v_i$ for $S^T \cdot S$. In the basis set spanned by the $u_i \otimes v_j$, we find that $S \cdot S^T \otimes I_d$ and $I_d \otimes S^T \cdot S$ are diagonal with $(\mu M) \nu_i$ as eigenvalues. Therefore, $\|S \cdot S^T \otimes I_d\|_\infty = (\mu M) \nu_m$ and $\|I_d \otimes S^T \cdot S\|_\infty = (\mu M)\nu_m$.
Moreover, observe that
\begin{equation}
(S^T \otimes S)^T \cdot (S^T \otimes S) \cdot(u_i\otimes v_j) = (S\cdot S^T \cdot u_i) \otimes (S^T \cdot S \cdot v_j) =(\mu M)^2 \nu_i \nu_j\  u_i \otimes v_j
\end{equation}
Therefore, $(S^T \otimes S)^T \cdot (S^T \otimes S) =(\mu M)^2 {\rm diag}(\nu_i \nu_j)$. In other words, $S^T \otimes S$ is the composition of an orthogonal operator (i.e. an isometry) and a diagonal matrix. Immediately, it follows that $\|S^T \otimes S\| \leq (\mu M)\nu_{m}$.

Compute the norm of equation (\ref{eq: necessary condition for eigenvalues of M0})
\begin{equation}
 |(\lambda + \eps \mu)(\lambda + 1)| \leq S'_m(|\lambda| + 3 \eps \mu) \nu_m .
\label{eq: second necessary condition for eigenvalues of M0}
\end{equation}
Define $f_\eps: \C \rightarrow \R$ such that $f_\eps(\lambda) = |(\lambda + \eps \mu)||(\lambda + 1)| - (|\lambda| + 3 \eps \mu) S'_m \nu_m$. We want to find a condition such that $f_\eps(\C_+)>0$, where $\C_+$ is the right half complex plane. This condition on $\eps,\ \mu,\ \nu_m,\ \text{ and }S'_m$ will be a sufficient condition for linear stability. Indeed, under this condition we can show that only eigenvalues with a negative real part can meet the necessary condition (\ref{eq: second necessary condition for eigenvalues of M0}). Complex number of the right half plane cannot be eigenvalues and thus the system is stable. The case $\eps = 0$ tells us that $f_0(\C_+)>0$ if $1>S'_m \nu_m$, compute
$$\frac{\partial f_\eps}{\partial \eps}(\lambda) = \mu(\Re(\lambda) + \mu \eps)\frac{|(\lambda + 1)|}{|(\lambda + \eps \mu)|} - 3 \mu S'_m \nu_m$$
If $1 \geq \eps \mu$, which is most probably true given that $\eps << 1$, then $\frac{|(\lambda + 1)|}{|(\lambda + \eps \mu)|} \geq 1$. Assuming $\lambda \in \C_+$ leads to:
$$\frac{\partial f_\eps}{\partial \eps}(\lambda) \geq \mu(\mu\eps - 3 S'_m \nu_m) \geq \mu(1 - 3 S'_m \nu_m) $$
Therefore, the condition $3S'_m \nu_m<1$, which implies $S'_m \nu_m < 1$, and leads to  $f_\eps(\C_+) > 0$.

 \subsection{A very short introduction to computational cohomology}
 \label{part: persistent cohomology}
In algebraic topology, topological spaces (which are continuous objects) can be classified by roughly counting their number of holes. This coordinate-invariant description of a topological space is called its homology (or cohomology, the difference between them is beyond the scope of this paper). In fact the homology can be summarized by giving the betti numbers of the topological state. The sequence of betti number is made of positive integers. The first three betti numbers have the following definition: the first is the number of connected components, the second is the number of two-dimensional or ``circular'' holes and the third is the number of 3-dimensional holes or ``voids''. See \cite{hatcher2002algebraic} for a more rigorous approach.

However, in the example of toroidal retinotopy (see section 5.3.2 and figure~\ref{fig: torus positions}). we are dealing with a discrete cloud of points. Therefore, one needs to extend the definition of the betti numbers to discrete objects in order to find the underlying topology of the space within which the points are distributed. This is called computational or persistent cohomology. One reconstructs the topological space by considering balls of a given radius centered on each point in the cloud. For each radius (the abscissa of the right picture of figure~\ref{fig: torus positions}), one can compute the betti numbers of the resulting topological space. A barcode graph, e.g. the right picture of figure~\ref{fig: torus positions}, is constructed by drawing a horizontal bar for each connected component, 2-dimensional hole or 3-dimenional void etc. This is done for a range of radii. Finding the persistent cohomology of a cloud of points consists in observing the set of betti numbers that are stable through a significantly large range of radii. One then assumes that the points most likely lie on a topological space of a given homology if the corresponding betti numbers are stable enough. See \cite{zomorodian2005computing} for details.

In section \ref{part: torus}, we used the Jplex software package of \cite{plex} to compute the barcodes of the points corresponding to learning from inputs uniformly distributed over a 2-dimensional torus. We used 200 landmarks spread according to the maxminlandmark method to build simplices in Jplex, which returned the maximum radius (beyond which all the betti numbers except the first vanish) we used in the simulation. In figure \ref{fig: torus positions}, we see that for a wide range of radii the triplet $(1,2,1)$ is stable. This corresponds to a 2 dimensional torus.

 \bibliographystyle{apalike}


\begin{thebibliography}{}

\bibitem[Amari, 1998]{amari1998natural}
Amari, S. (1998).
\newblock {Natural gradient works efficiently in learning}.
\newblock {\em Neural computation}, 10(2):251--276.

\bibitem[Amari et~al., 1992]{amari1992information}
Amari, S., Kurata, K., and Nagaoka, H. (1992).
\newblock {Information geometry of Boltzmann machines}.
\newblock {\em IEEE Transactions on Neural Networks}, 3(2):260--271.

\bibitem[Bartsch and Van~Hemmen, 2001]{bartsch2001combined}
Bartsch, A. and Van~Hemmen, J. (2001).
\newblock {Combined Hebbian development of geniculocortical and lateral
  connectivity in a model of primary visual cortex}.
\newblock {\em Biological Cybernetics}, 84(1):41--55.

\bibitem[Bi and Poo, 2001]{bi2001synaptic}
Bi, G. and Poo, M. (2001).
\newblock {Synaptic modification by correlated activity: Hebb's postulate
  revisited.}
\newblock {\em Annual review of neuroscience}, 24:139.

\bibitem[Bienenstock et~al., 1982]{bienenstock-cooper-etal:82}
Bienenstock, E., Cooper, L., and Munro, P. (1982).
\newblock Theory for the development of neuron selectivity: orientation
  specificity and binocular interaction in visual cortex.
\newblock {\em J Neurosci}, 2:32--48.

\bibitem[Borg and Groenen, 2005]{borg2005modern}
Borg, I. and Groenen, P. (2005).
\newblock {\em {Modern multidimensional scaling: Theory and applications}}.
\newblock Springer Verlag.

\bibitem[Bosking et~al., 1997]{bosking1997orientation}
Bosking, W., Zhang, Y., Schofield, B., and Fitzpatrick, D. (1997).
\newblock {Orientation selectivity and the arrangement of horizontal
  connections in tree shrew striate cortex}.
\newblock {\em Journal of neuroscience}, 17(6):2112.

\bibitem[Bressloff, 2005]{bressloff:05}
Bressloff, P. (2005).
\newblock Spontaneous symmetry breaking in self--organizing neural fields.
\newblock {\em Biological Cybernetics}, 93(4):256--274.

\bibitem[Bressloff et~al., 2001]{bressloff-cowan-etal:01}
Bressloff, P., Cowan, J., Golubitsky, M., Thomas, P., and Wiener, M. (2001).
\newblock Geometric visual hallucinations, euclidean symmetry and the
  functional architecture of striate cortex.
\newblock {\em Phil. Trans. R. Soc. Lond. B}, 306(1407):299--330.

\bibitem[Bressloff and Cowan, 2003]{bressloff-cowan:03}
Bressloff, P.~C. and Cowan, J.~D. (2003).
\newblock A spherical model for orientation and spatial frequency tuning in a
  cortical hypercolumn.
\newblock {\em Philosophical Transactions of the Royal Society B}.

\bibitem[Brewer, 1978]{brewer:78}
Brewer, J. (1978).
\newblock Kronecker products and matrix calculus in system theory.
\newblock {\em IEEE Transactions on Circuits and Systems}, 25(9).

\bibitem[Chklovskii et~al., 2002]{chklovskii2002wiring}
Chklovskii, D., Schikorski, T., and Stevens, C. (2002).
\newblock {Wiring optimization in cortical circuits}.
\newblock {\em Neuron}, 34(3):341--347.

\bibitem[Chossat and Faugeras, 2009]{chossat2009hyperbolic}
Chossat, P. and Faugeras, O. (2009).
\newblock {Hyperbolic planforms in relation to visual edges and textures
  perception}.
\newblock {\em PLoS Computational Biology}, 5(12):367--375.

\bibitem[Cohen and Grossberg, 1983]{cohen-grossberg:83}
Cohen, M. and Grossberg, S. (1983).
\newblock Absolute stability of global pattern formation and parallel memory
  storage by competitive neural networks.
\newblock In {\em IEEE Transactions on Systems, Man, and Cybernetics, SMC-13},
  pages 815--826.

\bibitem[Coifman et~al., 2005]{coifman2005geometric}
Coifman, R., Maggioni, M., Zucker, S., and Kevrekidis, I. (2005).
\newblock {Geometric diffusions for the analysis of data from sensor networks}.
\newblock {\em Current opinion in neurobiology}, 15(5):576--584.

\bibitem[Coombes, 2005]{coombes2005waves}
Coombes, S. (2005).
\newblock {Waves, bumps, and patterns in neural field theories}.
\newblock {\em Biological Cybernetics}, 93(2):91--108.

\bibitem[Cottet, 1995]{cottet1995neural}
Cottet, G. (1995).
\newblock {Neural networks: Continuous approach and applications to image
  processing}.
\newblock {\em Journal of Biological Systems}, 3:1131--1139.

\bibitem[Dayan and Abbott, 2001]{dayan-abbott:01}
Dayan, P. and Abbott, L. (2001).
\newblock {\em Theoretical Neuroscience : Computational and Mathematical
  Modeling of Neural Systems}.
\newblock MIT Press.

\bibitem[Degond and Mas-Gallic, 1989]{degond1989weighted}
Degond, P. and Mas-Gallic, S. (1989).
\newblock {The Weighted Particle Method for Convection-Diffusion Equations.
  Part 1: The Case of an Isotropic Viscosity}.
\newblock {\em Mathematics of Computation}, pages 485--507.

\bibitem[Dong and Hopfield, 1992]{dong1992dynamic}
Dong, D. and Hopfield, J. (1992).
\newblock {Dynamic properties of neural networks with adapting synapses}.
\newblock {\em Network: Computation in Neural Systems}, 3(3):267--283.

\bibitem[Edwards, 1996]{edwards1996approximation}
Edwards, R. (1996).
\newblock {Approximation of neural network dynamics by reaction-diffusion
  equations}.
\newblock {\em Mathematical methods in the applied sciences}, 19(8):651--677.

\bibitem[Faugeras et~al., 2008]{faugeras-grimbert-etal:08}
Faugeras, O., Grimbert, F., and Slotine, J.-J. (2008).
\newblock Abolute stability and complete synchronization in a class of neural
  fields models.
\newblock {\em SIAM J. Appl. Math}, 61(1):205--250.

\bibitem[F{\"o}ldi{\'a}k, 1991]{foldiak1991learning}
F{\"o}ldi{\'a}k, P. (1991).
\newblock {Learning invariance from transformation sequences}.
\newblock {\em Neural Computation}, 3(2):194--200.

\bibitem[Geman, 1979]{geman79}
Geman, S. (1979).
\newblock Some averaging and stability results for random differential
  equations.
\newblock {\em SIAM J. Appl. Math}, 36(1):86--105.

\bibitem[Gerstner and Kistler, 2002]{gerstner-kistler:02}
Gerstner, W. and Kistler, W.~M. (2002).
\newblock Mathematical formulations of hebbian learning.
\newblock {\em Biological Cybernetics}, 87:404--415.

\bibitem[Hatcher, 2002]{hatcher2002algebraic}
Hatcher, A. (2002).
\newblock {\em {Algebraic topology}}.
\newblock Cambridge Univ Pr.

\bibitem[Hebb, 1949]{hebb:49}
Hebb, D. (1949).
\newblock {\em The organization of behavior: a neuropsychological theory.}
\newblock Wiley, NY.

\bibitem[Hubel and Wiesel, 1977]{hubel77}
Hubel, D.~H. and Wiesel, T.~N. (1977).
\newblock Functional architecture of macaque monkey visual cortex.
\newblock {\em Proc. Roy. Soc. B}, 198:1--59.

\bibitem[Khalil and Grizzle, 1996]{khalil-grizzle:96}
Khalil, H. and Grizzle, J. (1996).
\newblock {\em Nonlinear systems}.
\newblock Prentice hall Upper Saddle River, NJ.

\bibitem[Kohonen, 1990]{kohonen1990self}
Kohonen, T. (1990).
\newblock {The Self-Organizing Map}.
\newblock {\em Proceedings of the IEEE}, 78(9).

\bibitem[Lawlor and Zucker, 2010]{zucker_abstract}
Lawlor, M. and Zucker, S. (2010).
\newblock {Third-Order Edge Statistics Reveal Curvature Dependency}.
\newblock In {\em Snowbird workshop on learning}.

\bibitem[Miikkulainen et~al., 2005]{miikkulainen-bednar-etal:05}
Miikkulainen, R., Bednar, J., Choe, Y., and Sirosh, J. (2005).
\newblock {\em Computational Maps in the Visual Cortex}.
\newblock Springer, New York.

\bibitem[Miller, 1996]{miller:96}
Miller, K. (1996).
\newblock Synaptic economics: competition and cooperation in synaptic
  plasticity.
\newblock {\em Neuron}, 17:371--374.

\bibitem[Miller and MacKay, 1996]{miller-mackay:96}
Miller, K. and MacKay, D. (1996).
\newblock The role of constraints in hebbian learning.
\newblock {\em Neural Comp}, 6:100--126.

\bibitem[Miller et~al., 1989]{miller89}
Miller, K.~D., Keller, J.~B., and Stryker, M.~P. (1989).
\newblock {Ocular dominance column development: analysis and simulation}.
\newblock {\em Science}, 245:605--615.

\bibitem[Oja, 1982]{oja:82}
Oja, E. (1982).
\newblock A simplified neuron model as a principal component analyzer.
\newblock {\em J. Math. Biology}, 15:267--273.

\bibitem[Ooyen, 2001]{ooyen:01}
Ooyen, A. (2001).
\newblock Competition in the development of nerve connections: a review of
  models.
\newblock {\em Network: Computation in Neural Systems}, 12(1):1--47.

\bibitem[Petitot, 2003]{petitot:03}
Petitot, J. (2003).
\newblock {The neurogeometry of pinwheels as a sub-Riemannian contact
  structure}.
\newblock {\em Journal of Physiology-Paris}, 97(2-3):265--309.

\bibitem[Sejnowski and Tesauro, 1989]{sejnowski1989hebb}
Sejnowski, T. and Tesauro, G. (1989).
\newblock {The Hebb rule for synaptic plasticity: algorithms and
  implementations}.
\newblock {\em Neural models of plasticity: Experimental and theoretical
  approaches}, pages 94--103.

\bibitem[Sexton and Vejdemo-Johansson, ]{plex}
Sexton, H. and Vejdemo-Johansson, M.
\newblock {JPlex simplicial complex library}.
\newblock \texttt{http://comptop.stanford.edu/programs/jplex/}.

\bibitem[Swindale, 1996]{swindale1996development}
Swindale, N. (1996).
\newblock {The development of topography in the visual cortex: a review of
  models}.
\newblock {\em Network: Computation in Neural Systems}, 7(2):161--247.

\bibitem[Takeuchi and Amari, 1979]{takeuchi-amari:79}
Takeuchi, A. and Amari, S. (1979).
\newblock Formation of topographic maps and columnar microstructures in nerve
  fields.
\newblock {\em Biological Cybernetics}, 35(2):63--72.

\bibitem[Tikhonov, 1952]{tikhonov1952systems}
Tikhonov, A. (1952).
\newblock {Systems of differential equations with small parameters multiplying
  the derivatives}.
\newblock {\em Matem. sb}, 31(3):575--586.

\bibitem[Verhulst, 2007]{verhulst2007singular}
Verhulst, F. (2007).
\newblock {Singular perturbation methods for slow--fast dynamics}.
\newblock {\em Nonlinear Dynamics}, 50(4):747--753.

\bibitem[Wallis and Baddeley, 1997]{wallis1997optimal}
Wallis, G. and Baddeley, R. (1997).
\newblock {Optimal, unsupervised learning in invariant object recognition}.
\newblock {\em Neural computation}, 9(4):883--894.

\bibitem[Zomorodian and Carlsson, 2005]{zomorodian2005computing}
Zomorodian, A. and Carlsson, G. (2005).
\newblock {Computing persistent homology}.
\newblock {\em Discrete and Computational Geometry}, 33(2):249--274.

\end{thebibliography}
\end{document}